\documentclass[a4paper,fleqn]{cas-dc}

\usepackage[utf8]{inputenc}
\usepackage[T1]{fontenc}

\newtheorem{theorem}{Theorem}
\newtheorem{definition}{Definition}
\newtheorem{lemma}{Lemma}

\usepackage[numbers]{natbib}

\usepackage{comment}

\usepackage{epstopdf}
\usepackage{subfigure}
\usepackage{algorithm}

\usepackage{array}
\usepackage{rotating}
\usepackage{multirow}
\usepackage{float}
\usepackage{url}
\usepackage{graphicx}

\usepackage{amsmath} 
\usepackage{natbib}

\usepackage{color,soul}

\usepackage{algpseudocode}
\usepackage{paralist}

\usepackage{forest}

\usepackage{xcolor}

\def\tsc#1{\csdef{#1}{\textsc{\lowercase{#1}}\xspace}}
\tsc{WGM}
\tsc{QE}

\begin{document}
\let\WriteBookmarks\relax
\def\floatpagepagefraction{1}
\def\textpagefraction{.001}

\sloppy

\shorttitle{ESBT: A Scalable and Deterministic Sequence CRDT for Distributed Collaborative Editing}    

\title [mode = title]{ ESBT: A Scalable and Deterministic Sequence CRDT for Distributed Collaborative Editing}

\author[drissaddress]{Moulay Driss Mechaoui}
\cormark[1]
\fnmark[1]
\ead{driss.mechaoui@univ-mosta.dz}
\affiliation[drissaddress]{organization={CSTL Lab, University of Mostaganem - Abdelhamid Ibn Badis},
           city={Mostaganem},
            country={Algeria}}

\author[imineaddress]{Abdessamad Imine}
\fnmark[2]
\ead{abdessamad.imine@loria.fr}
\affiliation[imineaddress]{organization={Universit\'e de Lorraine, CNRS, Inria, LORIA},
            city={F-54000 Nancy},
            country={France}}

\cortext[1]{Corresponding author}

\begin{abstract}
Modern collaborative editing systems require efficient mechanisms for managing concurrent updates across distributed replicas. Sequence Conflict-free Replicated Data Types (CRDTs) have become the de facto standard for supporting decentralized collaboration; they enable decentralized replicas to apply operations in arbitrary order while converging to a common state. Although existing sequence CRDTs guarantee Strong Eventual Consistency (SEC), they often suffer from uncontrolled identifier growth and increasing memory consumption during long-running, highly concurrent editing sessions, which severely limits their scalability and performance.
This paper presents the Extended Stern–Brocot Tree (ESBT), a mathematically grounded identifier allocation scheme for distributed collaborative text editing.
ESBT extends the classical Stern-Brocot Tree representation by combining bounded rational fractions ($f$), sequence numbers ($sn$), and sequence paths ($sc$) into a hierarchical identifier structure that provides a dense, deterministic, and compact identifier space. This design bounds identifier growth while preserving deterministic ordering and Strong Eventual Consistency. In addition, ESBT integrates a Red–Black tree document representation, enabling logarithmic-time ($O(\log n)$) insertion, deletion, and lookup operations.\\
Experimental evaluation using workloads of up to 100,000 concurrent operations generated across 50 collaborating sites shows that ESBT improves responsiveness by 28 to 88\% under pure insertions and 59 to 74\% under mixed insertion/deletion workloads, while reducing identifier memory consumption by 50 to 75\% in beginning and random insertion patterns compared with the best-performing baseline sequence CRDTs (Logoot and LSEQ). Under the adversarial middle-insertion workload (10,000 operations), ESBT further improves responsiveness by 86.53\% and reduces identifier size by 92.81\%. These results demonstrate that ESBT effectively addresses the principal scalability limitations of existing sequence CRDTs and provides an efficient and scalable foundation for large-scale collaborative editing systems.
\end{abstract}


\maketitle

\section{Introduction}
Collaborative editing platforms have become an essential component of modern distributed applications, enabling multiple geographically distributed users to update shared text, code, or graphical content simultaneously. To ensure high availability and low latency, each user maintains a local replica that can be updated independently without immediate synchronization across replicas. However, maintaining consistency across all replicas in a decentralized manner remains challenging. Beyond collaborative editing, maintaining deterministic ordering of concurrent operations is equally important in replicated data stores, distributed ledgers, replicated state machines, multiplayer game states, and other large-scale distributed systems.
Classical coordination mechanisms such as Lamport timestamps~\cite{lamport} and vector clocks~\cite{vector} can establish causality but fall short when it comes to enforcing a dense, deterministic order for concurrent operations. This limitation becomes highly evident in decentralized collaborative text editing. Without central coordination, replicas must independently resolve a consistent state, as multiple users simultaneously insert new characters directly between existing ones. 
Early collaborative editors addressed this challenge using Operational Transformation (OT)~\cite{Imi06}, which offers strong responsiveness through transformation functions that adjust the parameters (positions) of concurrent operations supporting arbitrary execution order. Although OT is effective for online collaboration with low divergence, its transformation overhead increases as replicas diverge, limiting its scalability in highly distributed and intermittently connected environments.\\ 
To avoid these scalability constraints, Conflict-Free Replicated Data Types (CRDTs)~\cite{logoot,crdt2024,ps-crdt} have been widely adopted as an alternative replication paradigm. CRDTs shift the focus from operation transformation to operation commutativity. Specifically, Sequence CRDTs guarantee that concurrent operations commute naturally. This enables replicas to execute updates independently and automatically achieve Strong Eventual Consistency (SEC)~\cite{sec17}. This property is preserved by attributing each element a unique and immutable identifier that defines its logical position within all replicas. Replicas can then merge concurrent insertions deterministically without centralized coordination, making sequence CRDTs the modern standard for decentralized collaborative architectures.\\
Despite their advantages, existing sequence CRDTs continue to exhibit important scalability limitations. Existing sequence CRDTs have evolved along two main lines: $(i)$ Tombstone-based approaches~\cite{woot, rga, geo_crdt26} preserve deleted elements (tombstones) to maintain order~\cite{crdt1}, which leads to unbounded state growth over time and costly garbage collection. $(ii)$ Identifier-allocation approaches~\cite{logoot, NedelecMMD13} avoid tombstones by assigning dense position identifiers. However, under unfavorable or highly localized insertion patterns, these identifiers may grow substantially, increasing both memory consumption and identifier comparison costs.\\ 
To the best of our knowledge, existing sequence CRDTs do not simultaneously ensure $(1)$ bounded identifier growth, $(2)$ tombstone-free deletion management, and $(3)$ lightweight synchronization without replica-sized causal metadata with baseline guarantees of deterministic ordering and Strong Eventual Consistency.  These limitations become increasingly significant in dynamic collaborative environments characterized by long editing sessions, intermittent connectivity, and users joining or leaving collaboration session at any time.\\
\noindent\textbf{Contributions}. To address these limitations, this paper proposes the Extended Stern-Brocot Tree (ESBT), a novel identifier allocation strategy for sequence CRDTs. ESBT extends the mathematical properties of the classical Stern-Brocot tree~\cite{stern1, stern2} to generate compact, deterministic, and densely ordered identifiers while preventing the unbounded growth of identifiers observed in existing identifier-allocation schemes, and without relying on tombstones. Combined with a self-balanced Red-Black tree, ESBT achieves logarithmic $O(\log n)$ insertion, deletion, and lookup operations, making it suitable for large-scale collaborative editing. Furthermore, ESBT introduces a lightweight dependency-based synchronization mechanism that preserves causal correctness using only direct operation dependencies, thereby avoiding replica-sized causal metadata while supporting dynamic participation in decentralized collaborative environments. The main contributions of this work are summarized as follows:
\begin{itemize}
\item  We propose ESBT, a novel Stern–Brocot-tree-based identifier allocation algorithm for sequence CRDTs that generates compact and deterministic position identifiers while preserving dense ordering.
\item We introduce a bounded identifier allocation mechanism that prevents uncontrolled identifier expansion, and supports tombstone-free deletion management.
\item We propose a lightweight synchronization protocol based on direct operation dependencies that preserves causal correctness without maintaining replica-sized causal metadata.
\item We formally prove that ESBT guarantees identifier uniqueness, deterministic total ordering, convergence, and Strong Eventual Consistency.
\item We experimentally evaluate ESBT against representative sequence CRDTs, including Logoot~\cite{logoot} and LSEQ~\cite{NedelecMMD13} under multiple allocation strategies and collaborative workloads. The results demonstrate lower memory consumption, more compact identifiers, and faster execution times under heavy concurrency. 
\end{itemize}
Although motivated by collaborative text editing, the proposed identifier allocation strategy can also be applied to a broader class of replicated ordered data structures, including hierarchical documents, XML, JSON, RDF graphs, Linked Data, and other Semantic Web data models~\cite{web_sem}.

\medskip
\noindent\textbf{Outline} The rest of this paper is organized as follows. Section~\ref{sec:background} reviews the properties of collaborative editing systems and introduces the necessary background. Section~\ref{sec:esbt} presents the Extended Stern–Brocot Tree (ESBT) model. Section~\ref{sec:coordination_model} describes the coordination model, while Section~\ref{sec:esbt_synch} details the ESBT synchronization layer. Section~\ref{sec:formal} formalizes the ESBT specification, and Section~\ref{sec:asymptotic} analyzes its asymptotic complexities. Section~\ref{sec:performance} reports the performance evaluation. Section~\ref{sec:related_work} discusses related work, and Section~\ref{sec:conclusion} concludes the paper. 

\section{Background}
\label{sec:background}
\subsection{Collaborative Editing Systems}\label{ces}
Collaborative editing systems typically rely on optimistic replication, which enables each user to independently modify a local replica of a shared document without immediate synchronization. This approach improves responsiveness and fault tolerance but poses the problem of ensuring that all replicas converge to an identical state despite simultaneous modifications.
Several correctness models have been proposed for collaborative editing. Among them, the CCI (Convergence, Causality, and Intention Preservation) model~\cite{Sun.ea:98} defines three essential requirements: all replicas must eventually converge to the same state, causal dependencies between operations must be preserved, and semantic intentions of users should be maintained whenever possible.
Modern decentralized collaborative systems commonly adopt Strong Eventual Consistency (SEC)~\cite{sec17}, which guarantees that replicas converge to the same state after all operations have been delivered, regardless of message ordering, provided that concurrent operations commute. SEC therefore provides the consistency model underpinning most CRDT-based collaborative editing systems.\\
The Stern–Brocot tree~\cite{stern1,stern2} provides a mathematical framework to generate an infinite, densely ordered set of rational numbers, making it an attractive foundation for identifier allocation in sequence CRDTs.

\subsection{Stern-Brocot Tree}
The Stern-Brocot tree is an infinite binary tree in which each node represents a distinct positive rational number in its reduced form that occurs exactly once~\cite{stern1,stern2}. The tree is built by recursively inserting the mediant between two adjacent fractions, where the mediant of two fractions \(\frac{a}{b}\) and \(\frac{c}{d}\) is \( \frac{a~+~c}{b~+~d} \). Starting from two sentinel nodes\(\frac{0}{1}\) and \(\frac{1}{0}\), their mediant is the root node \(\frac{1}{1}\), and the process continues recursively as depicted in Figure~\ref{fig:tree}. Each node has a well-defined level, which reflects its depth in the tree. The Stern-Brocot tree naturally sorts out the fractions in ascending order under an in-order traversal. By construction, every fraction created through the mediant is unique and preserves the numerical ordering of adjacent fractions. 

\begin{figure*}[ht!]
\centering
\includegraphics[scale=0.4]{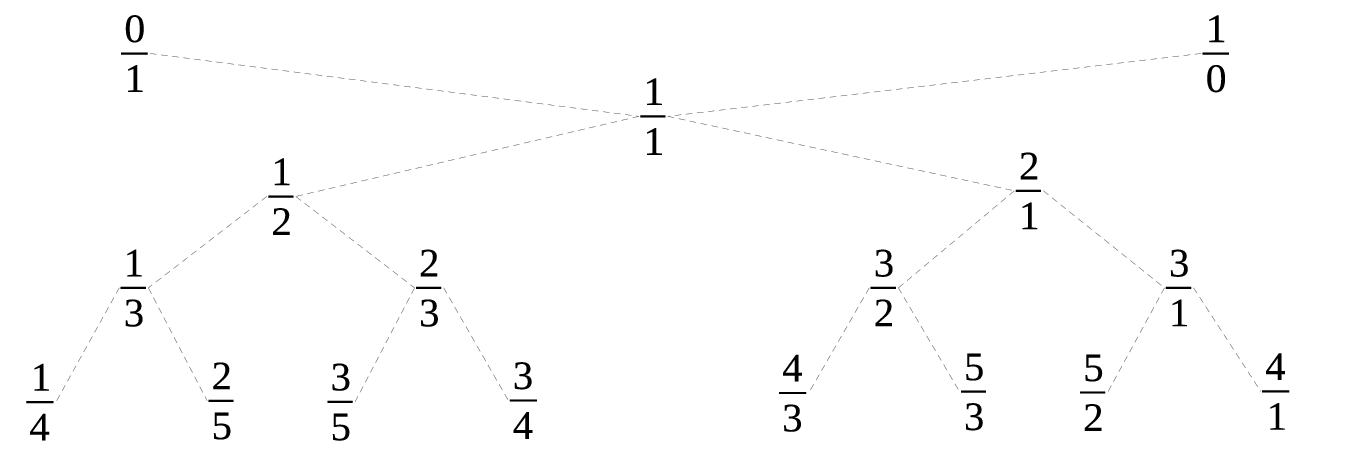}
\caption{The Stern–Brocot tree structure~\cite{stern1}.}
\label{fig:tree}
\vspace{-5mm}
\end{figure*}  

The Stern-Brocot tree has several properties that make it efficient that make it attractive for identifier allocation. First, each fraction is the mediant of its two parent fractions, making its value independent of other parts of the tree. Second, fractions within each level are always sorted by magnitude. These features allow for a targeted search, eliminating the need to compute the entire tree. Finally, the recursive mediant process generates an infinitely dense sequence of fractions, which ensures that new fractions can always be allocated between adjacent fractions. These properties make the tree a particularly interesting candidate for generating unique identifiers per site (fraction, site identifier) in collaborative editing systems.\\ 
Nevertheless, the classical Stern-Brocot tree is not sufficient to resolve concurrent insertions that concurrently generate the same rational fraction at different replicas. 
Additional mechanisms are therefore required to guarantee identifier uniqueness while preserving deterministic ordering in decentralized collaborative environments.

\section{Extended Stern-Brocot Tree}
\label{sec:esbt}
This section presents the Extended Stern–Brocot Tree (ESBT), a novel identifier allocation strategy for sequence CRDTs. ESBT extends the classical Stern-Brocot tree~\cite{stern1,stern2} by augmenting mediant-based rational identifiers with additional disambiguation parameters to guarantee uniqueness, deterministic ordering, and bounded identifier growth under concurrent insertions.

\subsection{ESBT Definition}
ESBT generalizes the mediant-based enumeration of the Stern-Brocot tree by constructing composite identifiers that remain totally ordered, immutable, and unique under concurrent updates. Each identifier encodes the logical position of a sequence element and serves as the basis for deterministic merging in collaborative editing.\\
The fundamental identifier in ESBT is the \textbf{Weight}, a composite identifier that uniquely determines the logical position of an element within global order. 

\begin{definition}[Weight]\label{def:weight}
An \emph{ESBT weight} is a quadruple
\(W = \langle f,sn,sc,\delta \rangle\) where:
\begin{itemize}
  \item \(f = \frac{p}{q} \in \mathbb{Q}^+ \) where gcd(p, q) =1, is an irreducible positive fraction (mediant) generated through the mediant operation of the Stern-Brocot tree and serves as the primary ordering key,
  \item \(sn \in \mathbb{Z}\) is a sequence number, used to distinguish successive insertions that share the same rational fraction. 
  \item \(sc \in \mathbb{N}^{\ast}\) sequence path is a finite lexicographically ordered sequence of non-negative integers used only when all available sequence numbers associated with the same fraction have been exhausted.
  \item \(\delta\) is the replica (site) identifier, used only as the final deterministic tie-breaker when all preceding ordering components are identical.
\end{itemize}
\end{definition}
Consequently, identifier ordering is determined primarily by the logical insertion position represented by $f$, then by the local insertion order encoded by $sn$, followed by the sequence path $sc$. The replica identifier \(\delta\) is consulted only as a final deterministic tie-breaker, thus preserving user intention while guaranteeing global uniqueness.

\begin{definition}[Total Ordering Relation]\label{def:lexorder}
For two weights \(W_i = \langle f_i, sn_i, sc_i, \delta_i \rangle \) and \(W_j = \langle f_j, sn_j, sc_j,\delta_j\rangle\),
\[
W_i < W_j \;\Leftrightarrow\;
\begin{cases}
\; f_i < f_j, \\[3pt]
\text{or }\; f_i = f_j \land sn_i < sn_j, \\[4pt]
\text{or }\; f_i = f_j \land sn_i = sn_j \land sc_i \prec sc_j, \\[4pt]
\text{or }\; f_i = f_j \land sn_i = sn_j \land sc_i = sc_j \land \delta_i \prec \delta_j
\end{cases}
\]
where \(\prec\) denotes lexicographic order on finite integer sequences.
\end{definition}

The hierarchical comparison ensures that the rational fraction remains the primary determinant of logical position. Additional components are consulted only when preceding components are identical, thereby preserving user-intended ordering while guaranteeing deterministic conflict resolution under concurrent insertions.

\subsection*{ESBT Tree Structure}
The ESBT is organized as an infinite, ordered, and rooted binary tree of immutable weights. Two sentinel weights delimit the entire ordering space:\\
$ W_{\mathrm{BEGIN}} = (\tfrac{0}{1}, 0, [0], \emptyset),\qquad W_{\mathrm{END}}   = (\tfrac{1}{0}, 0, [0], \emptyset). $  The fractions $\tfrac{0}{1}$ and $\tfrac{1}{0}$ represent the lower and upper boundaries of the rational ordering space, respectively, and $\emptyset$ denotes the null site identifier reserved exclusively for sentinel nodes.
The root node is defined as \( W_{root} =\left(\frac{1}{1}, 0, [0],\emptyset \right) \). \\
Unlike the classical Stern–Brocot tree, every node in the ESBT stores a complete weight rather than only a rational fraction, allowing deterministic ordering under concurrent insertions.
\subsection*{Structural Invariants}
\begin{definition}[Structural Invariants]
The ESBT satisfies the following invariants:
\begin{enumerate}
    \item \textbf{Immutability.}
    Every weight is immutable after creation and uniquely identifies the logical position of a sequence element.
    \item \textbf{Ordering.}
    For every node $W_P$, all weights stored in the left subtree are smaller than $W_P$, while all weights stored in the right subtree are greater than $W_P$, according to the total ordering relation defined in Definition~2.
    \item \textbf{Uniqueness.}
    No two distinct sequence elements may share the same weight. When identical fractions are generated concurrently, the additional ordering components $(sn,sc,\delta)$ guarantee uniqueness without modifying previously allocated identifiers.
\end{enumerate}

\end{definition}

\subsubsection*{Allocation Principle}
Given two adjacent weights
\[
W_L=\langle f_L, sn_L, sc_L, \delta_L\rangle,
\qquad
W_R=\langle f_R, sn_R, sc_R, \delta_R\rangle,
\]
where  \[ f_L=\frac{p_L}{q_L},  \qquad  f_R=\frac{p_R}{q_R}, \]
ESBT allocates a new identifier by computing the mediant
\[  f_M  =  \frac{p_L+p_R}{q_L+q_R}. \]
The newly generated weight $W_M$ is inserted between $W_L$ and $W_R$ as \( W_L < W_M < W_R \) thereby preserving the strict total ordering defined in Definition~2.

\subsubsection*{Bounded Allocation Policy}
To prevent unbounded identifier growth, ESBT introduces a configurable bound $D_{\max}$ on the numerator and denominator of newly generated mediants. When this bound is exceeded, allocation proceeds hierarchically through the disambiguation layers $sn$, $sc$, and finally $\delta$, thereby preserving uniqueness, deterministic ordering, and bounded identifier growth.

\subsection{Hierarchical Ordering Layers}
To guarantee unique and totally ordered identifiers under arbitrary concurrent insertions, ESBT organizes each weight into three hierarchical ordering layers. The fraction layer provides the primary logical position using the Stern–Brocot tree, the sequence number layer extends this ordering when fractional refinement reaches the bound \(D_{\max}\), and the sequence path layer resolves any remaining collisions lexicographically. Consequently, ESBT preserves deterministic ordering, guarantees uniqueness, bounds fraction growth, and expands identifiers only when required by concurrent conflicts.

\subsubsection{Fraction Layer}
The fraction layer provides the primary ordering component of an ESBT weight. It exploits the density of the Stern--Brocot tree to allocate new weights between existing neighbors elements without requiring pre-allocation or global coordination.
Each weight in this layer is represented by an irreducible fraction: \\
\(f = \frac{p}{q}, \qquad p, q \in \mathbb{N}, \; \gcd(p,q)=1, \) which corresponds to a unique position in the infinite Stern–Brocot tree. Ordering is immediate:\\
\(\frac{p_1}{q_1} < \frac{p_2}{q_2} \;\;\Leftrightarrow\;\; p_1 q_2 < p_2 q_1. \)\\
The generation of a new weight proceeds by computing the \emph{mediant} of the left and right neighbors:\\
\( \text{mediant}\left(\frac{a}{b}, \frac{c}{d}\right) =\frac{a + c}{b + d}.\) \\
The mediant fraction lies strictly between its parents, which preserves order without needing global coordination. Repeated mediant operations allow for an arbitrarily dense embedding of new weights between any two existing ones.

\paragraph{Boundedness Constraint.}
To prevent unlimited numerator or denominator growth, ESBT enforces a global threshold \(D_{\max}\).
Once this bound is reached, further fractional refinement is suspended and identifier allocation proceeds through the sequence number and sequence path layers.
This design bounds growth in the hierarchy of the fraction space while preserving the ability to generate infinite weights for each logical position.

\begin{lemma}[Bounded Fraction Property]
Let each weight in ESBT be defined as:
\vspace{-3mm}
\[ W = \langle f = \tfrac{p}{q},\, sn,\, sc,\, \delta \rangle \]
If the ESBT allocation enforces a bounded numerator or denominator constraint such that:
\vspace{-3mm}
\[
p \leq D_{\max} \quad \text{or} \quad q \leq D_{\max},
\]
then every fraction \( f \) generated by the allocator remains within a finite, dense, and totally ordered rational interval.
\end{lemma}

\subsubsection{Sequence Number Layer}
The fraction $f$ is not sufficient to guarantee uniqueness when multiple sites (replicas) concurrently insert elements at the same fractional position. 
The second layer, \emph{sequence number} ($sn$), provides a lightweight, integer-based disambiguation mechanism that ensures deterministic and bounded ordering among elements sharing the same fraction. Formally, for any two weights:
\[
W_i = \langle f_i, sn_i, sc_i, \delta_i \rangle
\quad \text{and} \quad
W_j = \langle f_j, sn_j, sc_j, \delta_j \rangle,
\]
if (\( f_i = f_j \)) then, their relative order is determined by comparing their local sequence numbers:
\[
sn_i < sn_j \; \Rightarrow \; W_i \prec W_j.
\]
This ordering rule ensures that weights originating from the same fractional interval are locally sequential and globally comparable across replicas.\\
Each site maintains a local \textbf{Tracker}, a hash map indexed by fraction values whose median exceeds $D_{max}$.

\begin{definition}[Site Tracker]\label{def:tracker}
Each site maintains a \emph{Tracker}:
\[\mathcal{T} : f \mapsto (sn_L, sn_R),\]
where \(sn_L\) and \(sn_R\) denote, respectively, the next \(sn\) to be allocated when allocating immediately to the left or right of the fraction \(f\).
The tracker allows controlled \(sn\) increment/decrement
instead of unbounded fraction refinement.
\end{definition}
The bounded fractions ($f$<$D_{max}$) remain untracked, ensuring that the structure remains lightweight and sparse. The sequence number layer resolves the conflict by monotonically increasing (or decreasing) \(sn\)  according to the allocation direction:

\begin{itemize}
\item \textbf{\textit{Left Allocation}}: When a new weight must precede an existing one that shares the same fraction \( f \), the site consults the left counter \( sn_L \) and assigns the next available lower integer:  
\(sn_{\text{new}} = sn_L - 1 \). 
\item \textbf{\textit{Right Allocation}}: In contract, when a new weight must follow within the same fraction, the right counter \( sn_R \) is incremented:  \(sn_{\text{new}} = sn_R + 1\).
\end{itemize}
The function ~\textsc{Create\_Weight} (see Algorithm~\ref{GId}) implements this policy by ensuring that, for any fraction \(f\) reaching the \(D_{\max}\) boundary, an infinite sequence of \(sn\) values is available without violating the order of neighboring weights. Each site allocates its local \(sn\) independently and  deterministically for a given fraction, ensuring that all replicas converge to a consistent global order that maintains strong eventual consistency.

\subsubsection{Sequence Path Layer}
\label{seq_path_layer}

The third layer is the \emph{sequence path} $sc$ where \( sc = [c_1, c_2, \dots, c_k] \) with \( c_i \in \mathbb{N} \). The \emph{sequence path} is a list of integers that acts as a deterministic tie-breaker in rare cases where two weights share the same fraction \( f \) and sequence number \( sn \).
By design, \(sc\) remains short or empty in most insertions; it is only created when \(sn\) alone cannot differentiate two concurrent insertions.\\
When \(f\) and \(sn\) are equal, the function \(\textsc{NewSeq}(sc_L, sc_R, base, DEPTH)\) (Algorithm~\ref{alg:newseq}) is invoked to compute a new sequence path strictly between the two neighbors \(sc_L\) and \(sc_R\).
The algorithm compares both paths digit by digit and attempts to locate the first depth where a numerical gap exists between the left value (\(lv\)) and right value (\(rv\)). 
If such a gap exists (\(interval = rv - lv - 1 > 0\)), the algorithm allocates a new digit at the midpoint of the interval:
\[ newVal = lv + \frac{interval}{2} + 1 \]
This midpoint allocation guarantees that the new path value lies strictly between its neighbors, ensuring deterministic ordering.\\
If the gap is exhausted (\(interval = 0\)), the common prefix is extended by appending \(lv\), and the algorithm descends to the next depth level.\\
If all depth levels are explored without finding any available gap (\(depth \geq DEPTH\)), a final fallback digit is appended using the site identifier: 
\vspace{-3mm}
\[tie = 1 + (siteId \bmod (base - 1)) \]
This ensures that each site still generates a unique sequence path, even under extreme concurrency.

\begin{algorithm}
\small
\caption{\textproc{newSeq}($left$, $right$, $base$, $\mathit{DEPTH}$)}
\label{alg:newseq}
\begin{algorithmic}[1]
\State $sc \gets []$
\State $depth \gets 0$
\While{true}
    \If{$depth < |left|$}
    \State $lv \gets left[depth]$
    \Else
    \State $lv \gets 0$
    \EndIf
    \If{$depth < |right|$}
    \State $rv \gets right[depth]$
    \Else
    \State $rv \gets base$
    \EndIf
    \State $interval \gets rv - lv - 1$
    \If{$interval > 0$}
        \State $newVal \gets lv + (\cfrac{interval}{2}) + 1$
        \If{$newVal \geq rv$}
            \State $newVal \gets rv - 1$
        \EndIf
        \State $sc.\text{add}(newVal)$
        \State \Return $sc$
    \Else
        \State $sc.\text{add}(lv)$
        \State $depth \gets depth + 1$
        \If{$depth \geq \mathit{DEPTH}$}
            \State $tie \gets 1 + (siteId \bmod (base - 1))$
            \State $sc.\text{add}(tie)$
            \State \Return $sc$
        \EndIf
    \EndIf
\EndWhile
\end{algorithmic}
\end{algorithm}

\subsubsection*{Example 1 (Gap available at low depth).}  
Consider \( sc_L = [3] \) and \( sc_R = [7] \), with \( base = 10 \), \( DEPTH = 3 \), and \( siteId = 2 \).\\  
At depth \(0\): \(lv = 3\), \(rv = 7\), giving \(interval = 3\).  
A midpoint is chosen: \(newVal = 3 + \frac{3}{2} + 1 = 5\).  
The resulting path is \( sc_{new} = [5] \), which lies strictly between the two neighbors.  
The process ends immediately since a valid midpoint was found.

\subsubsection*{Example 2 (No gap, deeper allocation required).}  Now take \( sc_L = [3] \) and \( sc_R = [4] \). At depth \(0\): \(lv = 3\), \(rv = 4\), resulting in \(interval = 0\). Since there is no gap, the algorithm appends \(3\) and moves to depth \(1\).  
At this level, the defaults are \(lv = 0\), \(rv = 10\), producing \(interval = 9\).  
The midpoint allocation gives \(newVal = 0 + \frac{9}{2} + 1 = 5\). The resulting sequence path is \( sc_{new} = [3, 5] \), which preserves strict order and stability. \\ 
If all depth levels are exhausted without finding a gap, the fallback rule appends a deterministic tie digit derived from \(siteId\), ensuring site-specific uniqueness.

\subsubsection{Sequence Path Order}
Since \(sc\) is totally ordered lexicographically, any two paths under the same fraction $f$ and sequence number $sn$ admit a deterministic comparison according the following definition.

\begin{definition}[Lexicographic Order on Sequence Paths]
\label{def:scorder}
Let $sc_i = [c_{i,0}, c_{i,1}, \dots, c_{i,m}]$ and 
$sc_j = [c_{j,0}, c_{j,1}, \dots, c_{j,n}]$ be two sequence paths. We define the lexicographic order $\prec$ as: 
\vspace{-3mm}
\[
sc_i \prec sc_j \;\;\Leftrightarrow\;\;
\exists\, d \;\; \big( c_{i,d} < c_{j,d} \;\;\land\;\;
\forall\, k < d,\, c_{i,k} = c_{j,k} \big).
\]
We say that, $sc_i$ precedes $sc_j$ if, at the first position $d$ where they differ, 
the digit of $sc_i$ is strictly smaller, while all earlier digits are equal. With $k$ an index for all earlier positions than $d$.
\end{definition} 
\vspace{-3mm}
In ESBT, most allocations are resolved at the fractional or sequence number layers, meaning that the vast majority of weights maintain an empty sequence path $sc = [0]$.
This layered design allows ESBT weights to achieve unlimited resolution while maintaining compactness, determinism, and performance efficiency across distributed replicas.

\begin{theorem}[Uniqueness]
For every insertion operation, \textsc{ESBT} generates a globally unique identifier through a hierarchical disambiguation process:
\begin{enumerate}
    \item a \emph{fractional mediant} $f_m$ is computed when \( f_m \leq D_{\text{max}} \);
    \item a \emph{sequence number} $sn$ is applied when mediants are exhausted \(f_m > D_{\text{max}}\) ;
    \item and a \emph{deterministic sequence-path}  $sc$ using \textsc{NewSeq} is applied when concurrent insertions still conflict.
\end{enumerate}
Consequently, the ordering relation \(<\) defines a total, replica-invariant order that remains consistent and preserved across all sites.
\end{theorem}

\begin{figure*}[ht!]
\centering
\includegraphics[scale=0.6]{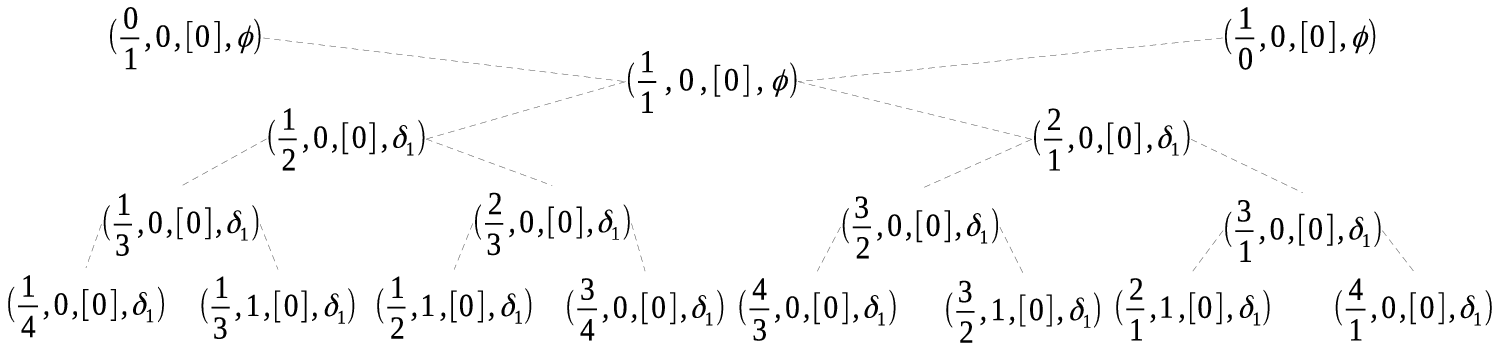}
\caption{The Extended Stern–Brocot tree structure (with $D_{max}$ = 5 ).}
\label{fig:esbttree}
\vspace{-5mm}
\end{figure*}  

\subsection{ESBT Allocation Function}
\label{GId}

\begin{algorithm}
\small
\caption{CREATE\_WEIGHT($\omega_1, \omega_2$)}
\label{alg:create-esbt-weight}
\begin{algorithmic}[1]
\State $\delta \gets$ site identifier
\State $base \gets$ available digit space at each depth
level
\State $depth \gets$ maximum number of levels
\State $num \gets \omega_1.p + \omega_2.p$
\State $den \gets \omega_1.q + \omega_2.q$
\State Let $\nabla \gets \gcd(num, den)$
\State $f_m.p \gets \frac{num}{\nabla}$,\quad $f_m.q \gets \frac{den}{\nabla}$
\State $f_m \gets \frac{f_m.p}{f_m.q}$
\State $D_{\max} \gets$ maximum $num$ and $den$ threshold

\If{$(f_m.p < D_{\max})$ \textbf{or} $(f_m.q < D_{\max})$}
    \State \Return $(f_m, 0, [0], \delta)$
\EndIf
\If{$\omega_1.f = \frac{0}{1}$}  
    \State $f_b \gets \omega_2.f$
\Else
    \State $f_b \gets \omega_1.f$
\EndIf
\State Tracker($f_b$) $\gets (sn_L, sn_R) $  
\If{$(f_b < \omega_2.f)$ \textbf{or} $((f_b = \omega_2.f) \textbf{ and } (sn_R < w_2.sn))$}
    \State $sn \gets sn_R + 1$
    \State $sc \gets \omega_1.sc$
    \State Update $sn_R \gets sn$ in Tracker
    \State \Return $(f_b, sn, sc, \delta)$
\ElsIf{ $(\omega_1.f < \omega_2.f)$ \textbf{ and } $( sn_L - 1 < \omega_2.sn)$}
    \State $sn \gets sn_L - 1$
    \State $sc \gets \omega_1.sc$
    \State Update $sn_L \gets sn$ in Tracker
    \State \Return $(f_b, sn, sc, \delta)$
\Else
    \State $sn \gets \omega_1.sn$
    \State $sc \gets$ \textproc{newSeq}($\omega_1.sc$, $\omega_2.sc$, $base$, $depth$)
    \State \Return $(f_b, sn, sc, \delta)$
\EndIf
\end{algorithmic}
\end{algorithm}


The \textsc{CREATE\_WEIGHT} function (Algorithm~\ref{alg:create-esbt-weight}) defines the core logic for allocating a new unique weight in ESBT. This function ensures consistent insertion order, bounded fraction precision, and convergence across replicas while preserving compact and ordered weights. The algorithm takes as input two adjacent weights $\omega_1$ and $\omega_2$, and returns a new ESBT weight of the form $(f, sn, sc, \delta)$.

\subsubsection*{Case 1: Mediant fits within $D_{max}$}
The function computes the mediant fraction $f_m$ between the two input fractions (lines 2--11):
\[f_m = \frac{\omega_1.p + \omega_2.p}{\omega_1.q + \omega_2.q} \] 
The resulting numerator and denominator are then reduced to the lowest terms using the greatest common divisor (GCD). This mediant provides a new rational value strictly between the two fractions, preserving the insertion order. If the simplified mediant numerator or denominator is less than $D_{\max}$, the function immediately returns a canonical weight with this fraction ($f_m$), a default sequence number $sn = 0$, and a default sequence path $sc = [0]$ (line~11). This case covers the majority of insertions in the sparse region of the weight space and avoids further disambiguation.\\
\textbf{Example} :  
Let \( w_1 = \left( \frac{1}{3}, 0, [0],\delta_1 \right), \quad w_2 = \left( \frac{1}{2}, 0, [0], \delta_1 \right),\) and $D_{\max} = 10$.\\  
The computed mediant is : $\frac{1 + 1}{3 + 2} = \frac{2}{5}$. 
Since ($2 < D_{\max}$) and ($5 < D_{\max}$), the mediant is valid. The function returns the new weight: ($\frac{2}{5}$, 0, [0],$\delta_1$).

\subsubsection*{Fallback Fraction Selection (Lines 13--16)}
If the mediant exceeds $D_{\max}$, the function falls back to the next allocation layer. A fallback fraction $f_b$ is selected from left weight $\omega_1.f$ unless $\omega_1$ is a sentinel (e.g., $\frac{0}{1}$), in which case $\omega_2.f$ is used. This ensures that disambiguation is always performed within a well-defined and non-sentinel range. This strategy biases allocations towards the left weight (which has a smaller fraction value), yet this consistent preference simplifies the algorithm's logic.
\subsubsection*{Case 2: Sequence Number Layer (Lines 18--28) }
The function uses a \textit{local Tracker} to maintain the left ($sn_L$) and right ($sn_R$) sequence number boundaries for each fallback fraction $\forall f_b : sn_L \leq sn \leq sn_R$.
It attempts to assign a new sequence number using the following rules:

\begin{itemize}
  \item \textit{Right Allocation}: If the fallback fraction $f_b$ is strictly less than $\omega_2.f$, or equal but with a smaller right $sn$, a new $sn$ is assigned as $sn_R + 1$ with $sc = [0]$. The right tracker is updated accordingly (lines 19--23).
  
  \item \textit{Left Allocation}: If $\omega_1.f$ less than $\omega_2.f$ and the left gap is sufficient ($\omega_2.sn < sn_L - 1$), a new $sn$ is assigned as $sn_L - 1$, again with $sc = [0]$. The left tracker is updated accordingly (lines 24--28).
\end{itemize}
These optimizations avoid generating deeper $sc$ paths when a free sequence number is available, ensuring better compactness of the weight and faster comparison.

\subsubsection*{Case 3: Sequence Path Layer (line~29--32)}
This final case is a fallback when the sequence number space between the relevant weights has been exhausted. The algorithm reuses the sequence number $sn_1$ and delegates the responsibility of creating a unique weight to the \textproc{newSeq} function (line~31) as presented in Sub-Section~\ref{seq_path_layer}.\\ 
The role of \textproc{newSeq} is to generate a new sequence path $sc_m$ that is lexicographically between $sc_1$ and $sc_2$.
This guarantees that weights can always be inserted between any two existing weights by extending the $sc$ path when needed. The function then returns the new weight with the fallback fraction $f_b$, the preserved sequence number, the newly generated $sc$ path, and the local site identifier $\delta$ (line~32).\\
By prioritizing mediant fractions, falling back to $sn$ when needed, and using $sc$ paths only as a last resort, the algorithm ensures both efficiency and global convergence.

\subsection{Illustrative Examples.}  
Consider a collaborative document represented as an Extended Stern–Brocot Tree (ESBT) with \(D_{\max} = 5\), \( base = 10 \), \( DEPTH = 3 \), and \( siteId = 2 \).\\ Suppose the current tree contains two adjacent weights at fraction layer:\\
\( w_1 = (\frac{1}{4}, 0, [], \delta_a) \) and 
\( w_2 = (\frac{2}{3}, 0, [], \delta_a) \).\\

\noindent \textbf{\textit{Situation 1: Fraction Layer Insertion.}} \\
When a user attempts to insert a new element between \(w_1\) and \(w_2\), the mediant is computed as
\(m = \frac{1+2}{4+3} = \frac{3}{7}.\) Since the denominator \(7 > D_{\max}\), the fraction layer cannot support this insertion, so the responsibility is delegated to the upper layers of the ESBT allocation strategy.\\

\noindent \textbf{\textit{Situation 2: Sequence Number Layer Insertion.}} \\ 
If a user attempts to insert between \( w_1 \) and \( w_2 \).  
The mediant is \( \frac{3}{7} \), since its denominator exceeds \( D_{\max} \), the algorithm falls back to the sequence number layer. \\ 
\textbf{\textit{Right Insertion.}} \\ 
Inserting after \( (\frac{1}{4}, 0, [0], \delta_a) \), the selected fallback fraction is \( f_b = \frac{1}{4} \).  
The right boundary is \( sn_R = 0 \), hence the new sequence number is \( sn = sn_R + 1 = 1 \).  The Tracker for \( \frac{1}{4} \) updates from \((0,0)\) to \((0,1)\), and the resulting weight is \( (\frac{1}{4}, 1, [0], \delta_a) \).\\  
If another insertion occurs after \( (\frac{1}{4}, 1, [0], \delta_a) \), the new sequence number is computed as \( sn = 2 \), the Tracker becomes \((0,2)\), and the new weight is
\( (\frac{1}{4}, 2, [0], \delta_a) \). \\
\textbf{\textit{Left insertion.}} Now consider inserting before \( w_1 = (\frac{1}{4}, 0, [], \delta_a) \).  
The left boundary is \( sn_L = 0 \), so the new sequence number is 
\( sn = sn_L - 1 = -1 \).  
The Tracker updates to \((-1,2)\), producing the weight \( (\frac{1}{4}, -1, [0], \delta_a) \).\\  
If we insert again before \( (\frac{1}{4}, -1, [0], \delta_a) \), then \( sn = -2 \), the Tracker updates to \((-2,2)\), and the new weight is: \( (\frac{1}{4}, -2, [0], \delta_a) \). \\ 
\textbf{\textit{Global order.}}  
The resulting order of weights is:\\
\((\frac{1}{4}, -2, [0], \delta_a) < (\frac{1}{4}, -1, [0], \delta_a) < (\frac{1}{4}, 0, [0], \delta_a) < (\frac{1}{4}, 1, [0], \delta_a) < (\frac{1}{4}, 2, [0], \delta_a) < (\frac{2}{3}, 0, [0], \delta_a)\).\\

\noindent\textbf{\textit{Situation 3: Sequence Path Layer Insertion.}}\\
If a new element is inserted between \((\frac{1}{4}, 0, [0], \delta_a)\) and \( (\frac{1}{4}, 1, [0], \delta_a)\)
there is no integer available for the new $sn$. In this case, the algorithm uses \textsc{NewSeq} for generating a refined $sc$. For instance, a possible $sc$ path for the new element is
\(w_{\text{new}} = (\frac{1}{4}, 0, [0,5], \delta_a),\)
which correctly orders it between the two existing weights:\\
\((\frac{1}{4}, 0, [0], \delta_a) < (\frac{1}{4}, 0, [0,5], \delta_a) < (\frac{1}{4}, 1, [0], \delta_a).\) 

\section{Distributed Coordination Model} 
\label{sec:coordination_model}
In our model, we consider an asynchronous distributed collaborative editing system composed of multiple sites connected through a network. We assume that the communication between sites is reliable, and the sites may join or leave the system at any time.  Updates are disseminated using an epidemic propagation mechanism, ensuring that each modification eventually reaches all sites, either directly or through intermediate relays~\cite{crdt2}.

\subsection{Shared Document Representation}
Each replica stores the shared document as a self-balancing Red--Black tree~\cite{cormen2022}, in which nodes are ordered according to their ESBT weights. This balanced binary search tree guarantees worst-case $O(\log n)$ time complexity for insertion, deletion, and lookup operations while maintaining a balanced height regardless of the editing workload.\\
Each document node is represented as the pair
\(\langle W,C\rangle, \)  where $W$ is an ESBT weight (Definition~\ref{def:weight}) that uniquely identifies the logical position of the element, and $C$ denotes the associated content. The content may represent a character, paragraph, list item, JSON object, XML node, or any other structured element, enabling ESBT to support both linear and hierarchical collaborative data. The document is bounded by two immutable sentinel nodes
\(\langle W_{\mathrm{BEGIN}},\bot\rangle
\qquad\text{and}\qquad
\langle W_{\mathrm{END}},\top\rangle,
\) where $\bot$ and $\top$ are reserved symbols denoting the lower and upper boundaries of the document, respectively. These sentinel nodes are system-defined and are never created, modified, or deleted by user operations.\\
This representation combines the deterministic ordering of ESBT weights with the logarithmic-time guarantees of Red--Black trees, providing an efficient foundation for scalable collaborative editing. Since the ordering relation is total, the logical position of each document element is uniquely determined, and weights released after deletions can be safely reallocated without compromising uniqueness or consistency.

\subsection{Editing Operations}
Users modify the state of the shared document through two primitive editing operations:
\noindent (i) \texttt{INS}$(\omega_{new},\, e)$, which inserts a new element $e$ into the document by associating a new weight $\omega_{new}$.
\noindent (ii) \texttt{DEL}$(\omega_{exist} )$ which removes the element identified by the existing weight $(\omega_{exist} )$. Unlike tombstone-based approaches, deleted weights are released and may subsequently be reused by the ESBT allocation protocol without violating uniqueness or consistency. Multiple replicas may execute editing operations concurrently. Operations are applied locally as soon as they are generated and propagated asynchronously to the remote replicas. Since ESBT weights are globally unique and totally ordered, all replicas deterministically converge to the same document state regardless of the order in which concurrent operations are received.
The unique ESBT weight $\omega$ is generated inside the interval $W_{\mathrm{BEGIN}}$ and $W_{\mathrm{END}}$. Let "$\prec$" be a order relation between position weights. We use $\omega$, $\omega'$, $\omega_1$, $\omega_2$, $\ldots$, to denote all ESBT weights. To insert an element between two ESBT weights $\omega_1$ and $\omega_2$, such that $\omega_1 \prec \omega_2$,
simply requires generating a new weight $\omega_{new}$ such that:
$\omega_1 \prec \omega_{new} \prec \omega_2$ using \textproc{CREATE\_WEIGHT} presented in Sub-Section~\ref{GId}.

\section{ESBT Synchronization Layer} 
\label{sec:esbt_synch}
This section introduces the synchronization layer of ESBT, which coordinates updates between replicas in a distributed environment. This layer resolves consistency issues, formalizes the causal relationship between operations, and presents a lightweight control procedure to ensure consistency with minimal communication and storage overhead.

\subsection{Consistency Issues} 
\label{consistency_issue}
Even in collaborative editing systems, naïvely exchanging operations can lead to inconsistencies if causality and uniqueness are not carefully preserved. The following representative scenarios motivate the synchronization mechanism adopted by ESBT. 

\begin{figure*}[hbt!]
  \centering
  \subfigure[Causal dependency issue between insertion and deletion.]{
    \includegraphics[scale=0.45]{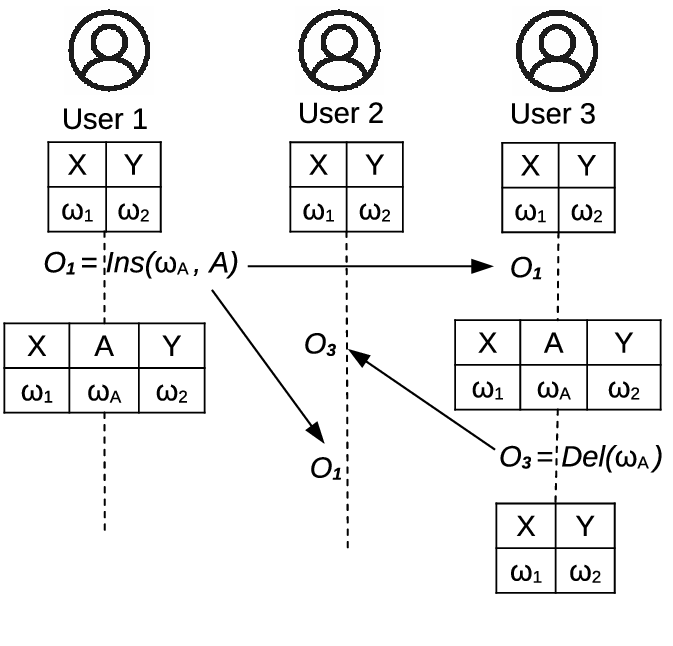}
    \label{subfig:issue1}
  }
  \quad
  \subfigure[Concurrent deletions issue of the same elements.]{
    \includegraphics[scale=0.45]{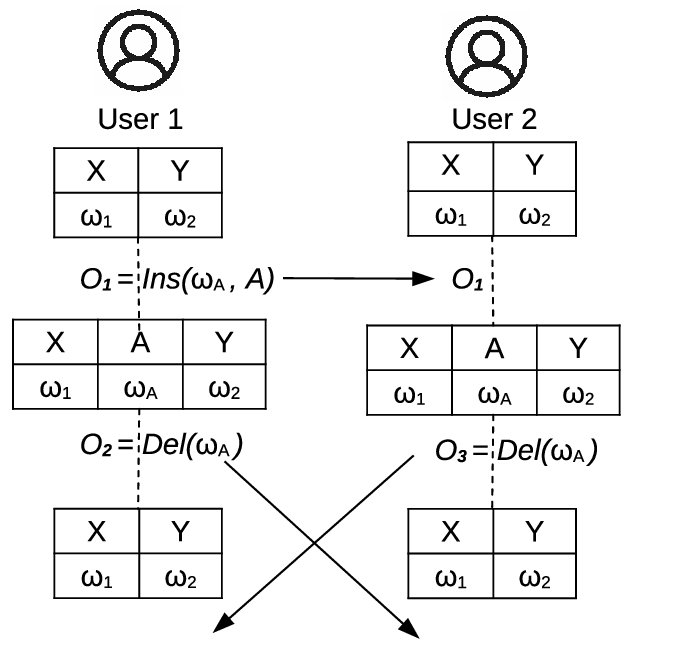}
    \label{subfig:issue2}
  }
  \quad
  \subfigure[problem of reusing the same weight for distinct elements.]{
    \includegraphics[scale=0.45]{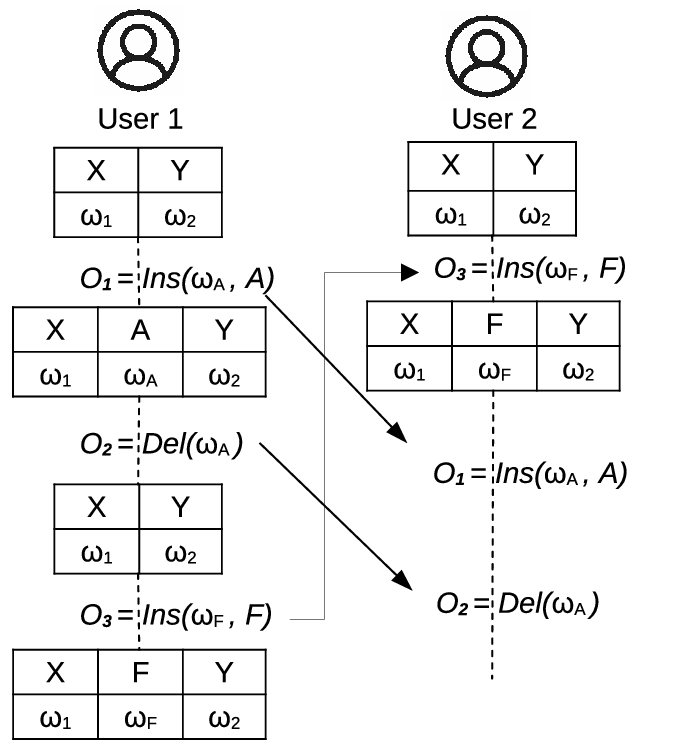}
    \label{subfig:issue3}
  }
  \caption{Consistency Issues Scenarios}
  \label{fig:exp}
\end{figure*}

\subsubsection*{Scenario 1: Causal Dependency Between Insertion and Deletion}
Consider a shared document initially containing the text "XY" shown in Figure~\ref{subfig:issue1}. User~1 inserts 'A' between 'X' and 'Y', while User~3 subsequently deletes 'A' after receiving the insertion. If another replica receives the deletion before the corresponding insertion due to message reordering, the deletion cannot be safely executed, since the referenced weight does not yet exist. Therefore, to ensure causal dependency, deletion must remain pending until its corresponding insertion has been integrated.
\[Ins(\omega_{A}, A) \;\rightarrow\; Del(\omega_{A}) \]
Therefore, User~2 delays the execution of $O_{2}$ until $O_{1}$ has been received and applied, ensuring consistent convergence of all replicas.

\subsubsection*{Scenario 2: Concurrent deletions of the same elements}
Figure~\ref{subfig:issue2} illustrates two users concurrently deleting the same element 'A'. When one deletion reaches a replica where 'A' has already been removed, executing the operation again would be redundant. To address this issue, each user maintains a local $Log$ of delete operations.
Incoming delete operations targeting previously deleted weights are recognized and safely ignored, guaranteeing idempotent deletion and replica convergence.

\subsubsection*{Scenario 3: Reinsertion at the Same Position}
Figure~\ref{subfig:issue3} illustrates a more subtle case. User~1 inserts 'A', deletes it, and later inserts 'B' at the same logical position. Because both insertions are computed from the same neighboring weights, they may receive the same weight. If replicas observe these operations in different orders, they cannot determine which insertion a deletion refers to.
To distinguish reused weights, each insertion is associated with a monotonically increasing counter $c$. Weight identifiers are therefore represented as:
\[
W = \{\,(\omega, c) \mid c \in \mathbb{N}\,\}
\]
where each pair is considered distinct according to
\[
(\omega_{1},c_{1}) \neq (\omega_{2},c_{2})
\quad \text{iff} \quad
\begin{cases}
\omega_{1} \neq \omega_{2}, & \text{or} \\
\omega_{1} = \omega_{2} \text{ and } c_{1} \neq c_{2}.
\end{cases}
\]
The counter $c$ is used only to identify successive insertions sharing the same weight; it does not affect the document ordering. Consequently, reused weights remain unique even under asynchronous message delivery.

\subsection{ESBT Causal Dependencies} 
\label{causal-dep}
ESBT is compatible with standard causal delivery protocols, including causal broadcast~\cite{causal-broadcast} and probabilistic causal broadcast~\cite{proba-broadcast}. 
These mechanisms correctly preserve causality, but often rely on replica-sized metadata, such as vector clocks, whose storage and communication costs increase with the number of participating replicas.
To avoid this overhead, ESBT associates each locally generated insertion with a lightweight \emph{counter} (denoted $c$), maintained in a small per-replica map. Every deletion operation carries the \emph{counter} of the insertion on which it depends, thereby establishing an explicit dependency between the two operations. Unlike vector clocks, this dependency information is independent of the ESBT ordering components $(f,sn,sc)$ and therefore does not affect identifier allocation or document ordering.\\
This dependency model allows replicas to correctly distinguish between successive insertions that may reuse the same ESBT weight after deletion. Consequently, deletion operations are always applied to their corresponding insertions, even in the presence of message reordering, duplication, or temporary network partitions. Since each operation carries only a single integer \emph{counter}, the synchronization overhead remains constant and independent of the number of replicas.

\subsubsection{Causal Dependency Between Operations}
To explicitly capture the causal relationship between insertions and deletions, ESBT associates every insertion with a local counter $c$. This counter uniquely identifies the insertion at a given replica and is carried unchanged by the corresponding deletion operation.

\begin{definition}[Causal Dependency]
Let \(\text{Ins}(\omega, e, c)\) denote an insertion of element \(e\) at weight \(\omega\) with counter \(c\),  
and let \(\text{Del}(\omega, c)\) denote a deletion referring the insertion with the same weight \(\omega\) and counter \(c\).\\  
We say that \(\text{Del}(\omega, c)\) is \textbf{causally dependent} on \(\text{Ins}(\omega, e, c)\), denoted by :
\[
\text{Ins}(\omega, e, c) \;\to\; \text{Del}(\omega, c)
\]
iff the deletion references the insertion identified by the same pair $(\omega,c)$.
\end{definition}

A deletion operation is applied only after its corresponding insertion has been integrated locally. If a deletion arrives before the referenced insertion due to message reordering, it is temporarily buffered until the dependency is satisfied.\\
This explicit dependency model ensures that every deletion removes exactly the intended insertion, even when ESBT weights are subsequently reused. Consequently, causal correctness is preserved without introducing replica-sized causal metadata or modifying the ESBT ordering layers.

\subsection{Lightweight Synchronization Procedure} 
\label{ccp}
Algorithm~\ref{alg:control_concurrency} presents the synchronization procedure executed independently at each replica. The protocol combines local-first execution with the ESBT identifier allocation scheme introduced in Section~\ref{GId} and the dependency model presented in Section~\ref{causal-dep}. As a result, replicas can process updates immediately while preserving deterministic convergence under asynchronous communication.

\subsubsection*{Generation of Local Operations}
Locally generated operations are executed immediately on the local replica before being propagated asynchronously to the remote participants.
For an insertion, the procedure allocates a new ESBT weight $\omega_e$ using the \textsc{Create\_Weight} algorithm and increments the local insertion counter $c$. The pair $(\omega_e,c)$ is stored in a local \textit{CounterMap}, establishing a persistent association between the inserted element and its dependency information. This association allows to perform the local insertion operation on any remote site, regardless of their reception order, thus avoiding the inconsistency illustrated in \textit{\textbf{Scenario~3}} of Sub-Section~\ref{consistency_issue}.\\ 
For deletion, the procedure retrieves the counter associated with the referenced insertion from \textit{CounterMap} (line~27) and generates the operation $\mathrm{Del}(\omega,c)$. The deletion is also stored in a local log ($L$) to avoid ambiguities when previously released weights are later reused. This log prevents the issue described in \textit{\textbf{Scenario~2}} of Sub-Section~\ref{consistency_issue}.\\ 
Finally, every generated insertion or deletion is disseminated using the underlying reliable broadcast mechanism, ensuring that all replicas eventually receive and integrate the operation.

\begin{algorithm}
\small
\caption{ESBT Control Concurrency Algorithm}
\label{alg:control_concurrency}
\begin{algorithmic}[1]

\Procedure{Main}{}
    \State \Call{Initialize}{}
    \While{running}
        \If{there is a local input $Op$}
            \State \Call{GenerateOperation}{$Op$}
        \Else
            \State \Call{ReceiveOperation}{}
            \State \Call{IntegrateRemoteOperation}{}
        \EndIf
    \EndWhile
\EndProcedure

\Procedure{Initialize}{}
    \State $Q \gets [\ ]$  \Comment{Queue of pending operations}
    \State $L \gets [\ ]$  \Comment{Delete Log}
    \State $S \gets S_0$   \Comment{Document state}
    \State $CounterMap \gets \{\}$ \Comment{Mapping: weight → counter}
    \State $Counter \gets 0$
    \State $\delta \gets$ local site identifier
\EndProcedure

\Procedure{GenerateOperation}{$Op$}
    \If{$Op$ is \textsc{Ins} between $\omega_i$ and $\omega_{i+1}$}
        \State $\omega_e \gets$ \Call{Create\_Weight}{$\omega_i, \omega_{i+1}$}
        \State $Counter \gets Counter + 1$
        \State $CounterMap[\omega_e] \gets Counter$  
        \State $Op \gets \textsc{Ins}(\omega_e, e, Counter)$
    \Else
        \State $Op \gets \textsc{Del}(\omega_e, CounterMap[\omega_e])$
        \State $L \gets L + Op$  
    \EndIf
    \State $S \gets$ \Call{Apply}{$Op, S$} 
    \State Broadcast $Op$ to other users
\EndProcedure

\Procedure{\textproc{ReceiveOperation}}{}
    \If{there is an operation $Op$ from the network}
        \State $Q \gets Q + Op$
    \EndIf
\EndProcedure

\Procedure{\textproc{IntegrateRemoteOperation}}{}
    \ForAll{$Op \in Q$ such that \Call{IsCausallyReady}{$Op$}}
        \State $Q \gets Q - Op$
        \State $S \gets$ \Call{Apply}{$Op, S$}
        \If{$Op$ is \textsc{Del}}
            \State $L \gets L + Op$  
        \EndIf
    \EndFor
\EndProcedure

\Function{\textproc{IsCausallyReady}}{$Op$}
    \If{$Op$ is \textsc{Del}($\omega_e, c$)}
        \If{$\omega_e \in S$} 
            \State \Return \textbf{true} \Comment{Insertion exists: Apply}
        \ElsIf{$(\omega_e, c) \in L$}
            \State \Return \textbf{false} \Comment{Already deleted: Ignore}
        \Else
            \State \Return \textbf{false} \Comment{Insertion not yet arrived: Wait}
        \EndIf
    \Else
        \State \Return \textbf{true} \Comment{Insertions are always Ready}
    \EndIf
\EndFunction

\end{algorithmic}
\end{algorithm}
\medskip

\subsubsection*{Integration of Remote Operations}
Each replica maintains a pending queue $Q$ that contains remote operations whose causal dependencies have not yet been satisfied. Upon reception, an operation generated at another replica is inserted into $Q$ and processed according to its type. 
Insertion operations are immediately integrated into the local document since they have no unresolved dependencies. In contrast, a deletion operation $\mathrm{Del}(\omega,c)$ is applied only if its corresponding insertion, identified by the pair $(\omega,c)$, has already been integrated locally. Otherwise, the deletion remains buffered in $Q$ until the required insertion is received, thus avoiding the issue illustrated in \textit{\textbf{Scenario~1}} of Sub-Section~\ref{consistency_issue}. \\
Whenever an operation becomes causally ready, the procedure \textsc{Apply}$(Op,S)$ is invoked to update the local document state $S$. This dependency-driven integration guarantees that every deletion is applied to its intended insertion despite message reordering, duplication, or temporary network partitions, while avoiding replica-sized causal metadata.

\section{Formal Model of ESBT}
\label{sec:formal}
To reason about correctness and convergence, we formalize the Extended Stern–Brocot Tree as an ordered set of \emph{weights}. 
We state a series of theorems showing that ESBT preserves a strict total order, generates globally unique weights, and converges under arbitrary operation delivery.  

\begin{theorem}[Mediant Constraint]\label{th:mediant}
Let \(D_{\max} \in \mathbb{N}\) be the predefined bound for both numerator and denominator.
For any attempted insertion between
\(W_L = (f_L, sn_L, sc_L, \delta_L)\) and
\(W_R = (f_R, sn_R, sc_R, \delta_R)\),
compute the mediant
\[
f^{\ast} = \frac{p_L + p_R}{q_L + q_R},
\qquad f_L = \frac{p_L}{q_L},\; f_R = \frac{p_R}{q_R}.
\]
If \(\max(\mathrm{num}(f^{\ast}),\mathrm{den}(f^{\ast})) \le D_{\max}\),
assign \(f^{\ast}\) as the fraction for the new weight.
Otherwise, ESBT handles this by:
\begin{enumerate}
  \item Allocating an unused sequence number \(sn\) from the site’s tracker,
  \item On concurrent collisions, generating a finite sequence path
        \(sc = \textsc{NewSeq}(sc_L, sc_R, base, DEPTH)\).
\end{enumerate}
Consequently,
\[ \exists!\, W = (f, sn, sc, \delta)\;\;\text{such that}\;\; W_L < W < W_R, \] 
while the numerator and denominator of the fraction component remain bounded by $D_{\max}$.
\end{theorem}

\noindent\textbf{Proof Sketch.}
If the mediant fraction satisfies the bound $D_{\max}$, it is uniquely determined by the Stern--Brocot tree and lies strictly between the neighboring fractions. Otherwise, ESBT deterministically refines the ordering using the sequence number $sn$ and, if necessary, the sequence path $sc$. Since each refinement preserves the ordering relation, a unique weight is always generated without exceeding the predefined fraction bound.

\begin{theorem}[Strict Total Order]\label{th:order}
The relation \(<\) from Definition~\ref{def:lexorder} forms a strict total order
on the set of ESBT weights.
\end{theorem}

\noindent\textbf{Proof Sketch.}
The ordering relation compares weights hierarchically according to $(f,sn,sc,\delta)$. Each component is itself totally ordered: fractions by rational order, sequence numbers by integer order, sequence paths lexicographically, and site identifiers deterministically. Consequently, every pair of distinct weights is comparable, establishing a strict total order.

\begin{theorem}[Weight Uniqueness]\label{th:unique}
Every insertion operation generates a unique ESBT weight.
No two distinct insertions can produce the same
\((f,sn,sc,\delta)\).
\end{theorem}

\noindent\textbf{Proof Sketch.}
The fraction component uniquely identifies a position whenever a new mediant can be allocated. Concurrent fraction collisions are resolved using the sequence number, followed by the sequence path when necessary, while the site identifier provides a final deterministic tie-breaker. Therefore, no two insertions can generate identical ESBT weights.

\begin{theorem}[Convergence]\label{th:convergence}
Given a finite set of insertion operations delivered to all replicas, ESBT ensures every replica converges to the same sequence of weights ordered by Definition~\ref{def:lexorder}.
\end{theorem}

\noindent\textbf{Proof Sketch.}
Weight generation is deterministic and independent of message delivery order. Since every replica eventually receives the same set of operations and applies the same total ordering relation, all replicas construct an identical ordered sequence. Therefore, ESBT guarantees deterministic convergence and satisfies Strong Eventual Consistency.

Theorems~\ref{th:mediant} and \ref{th:convergence} establish the fundamental correctness properties of ESBT. Theorem~\ref{th:mediant} guarantees bounded fraction allocation through hierarchical disambiguation, Theorem~\ref{th:order} proves that ESBT weights form a strict total order, Theorem~\ref{th:unique} ensures global weight uniqueness, and Theorem\ref{th:convergence} establishes deterministic convergence across replicas. Together, these properties provide the formal foundation for ESBT and ensure Strong Eventual Consistency in distributed collaborative editing.
\section{Asymptotic Complexity}
\label{sec:asymptotic}
The computational cost of ESBT can be analyzed by distinguishing identifier allocation from document integration.

\paragraph{Local Weight Allocation.}
Generating a new ESBT weight consists of computing a mediant fraction between two adjacent identifiers and, when the bound $D_{\max}$ is reached, activating the successive disambiguation layers ($sn$, $sc$, and finally $\delta$). Since each allocation performs a bounded number of arithmetic and comparison operations per refinement level, the identifier generation cost grows logarithmically with the search depth. Moreover, the bounded allocation policy limits the size of each identifier component, yielding constant memory overhead per weight.

\paragraph{Document Integration.}
The shared document is indexed by a Red--Black tree ordered according to ESBT weights. Consequently, searching for an existing weight, inserting a new element, and deleting an existing element all require \(O(\log n),\) where $n$ denotes the number of elements of the document . Therefore, both locally generated and remotely received operations are integrated in logarithmic time.
In general, ESBT combines bounded identifier allocation with logarithmic-time document operations, making it suitable for large-scale collaborative editing with high concurrency.

\section{Performance Evaluation}
\label{sec:performance}
This section evaluates the proposed ESBT approach from two complementary perspectives: (i) the influence of its allocator parameters on identifier growth and memory consumption, and (ii) its performance relative to representative sequence CRDTs under standard collaborative editing workloads.\\
The first experiment investigates the impact of the ESBT allocator parameters (\textit{base} and \textit{depth}) on identifier size and memory usage. The second experiment compares ESBT with two representative variable-size identifier sequence CRDTs, namely Logoot~\cite{logoot} and LSEQ~\cite{NedelecMMD13}, using the Boundary, Random, and Mixed allocation strategies. The evaluation considers four representative insertion workloads: beginning, end, middle, and random.\\
All experiments were executed on a machine with
Processor 12$^{th}$ Gen Intel® Core™ i9-12900K × 24,  OS Name Ubuntu 22.04.5 LTS, and Memory 32.0 GiB. All algorithms were implemented in Java.

\subsection{Impact of ESBT Allocator Parameters}
This experiment evaluates the influence of the ESBT allocator parameters on weight growth and memory consumption. Since the sequence path ($sc$) is the only variable-length component of an ESBT weight, the evaluation focuses on the behavior of the \textsc{NewSeq} function under the worst-case insertion pattern, namely repeated middle insertions. This workload maximizes identifier refinement by repeatedly extending the sequence path and therefore provides an upper bound on identifier growth.

\begin{table}[ht]
	\centering
	\caption{ESBT weight size (in MB) under varying base configurations and worst-case insertion workload.}
	\begin{tabular}{lrrr}
		\hline
		\textbf{Base} & \textbf{1\,000 ops} & \textbf{10\,000 ops} & \textbf{100\,000 ops} \\
		\hline
		$2^{8}$   & 0.490 & 49.900 & 3202.320 \\
		$2^{10}$  & 0.399 & 39.992 & 3999.920 \\
		$2^{16}$  & 0.262 & 25.129 & 2501.299 \\
		$2^{31}-1$ & 0.152 & 13.144 & 1292.735 \\
		\hline
	\end{tabular}
	\label{tab:esbt_middle_insertion}
\end{table}

\subsubsection{Experiment Setup}
The evaluation consists of two phases, each comprising up to 100\,000 insertion operations.
In the first phase, four branching factors (bases) were evaluated: \(\{2^{8},\,2^{10},\,2^{16},\,2^{31}-1\}.\)
These values span small, medium, and large allocation spaces and were used to measure the total memory footprint after 1\,000, 10\,000, and 100\,000 insertions. 
During this phase, the maximum depth was fixed to $2^{16}$ in order to isolate the influence of the branching factor.\\
In the second phase, the best-performing branching factor identified in Table~\ref{tab:esbt_middle_insertion} was retained while varying the maximum depth. The objective was to determine the smallest depth that minimizes the sequence-path length without compromising scalability or identifier allocation.

\subsubsection{Results}
\noindent\textbf{\textit{Phase 1}}.
Table~\ref{tab:esbt_middle_insertion} reports the memory consumption of ESBT under the worst-case middle-insertion workload for different branching factors (bases). The results show a clear inverse relationship between the base value and memory usage. Increasing the base enlarges the available allocation space between adjacent weights, thereby reducing the frequency of sequence path ($sc$) extensions.\\
After 100\,000 insertions, the configuration with base $2^{31}-1$ required 1\,293~MB, compared with 3\,202~MB for base $2^{8}$, corresponding to a memory reduction of approximately 59.6\%. The intermediate configurations ($2^{10}$ and $2^{16}$) exhibit the same trend, indicating that larger branching factors produce shallower identifiers and more compact memory representations. 
Based on these observations, the remaining experiments use a branching factor of $2^{31}-1$, which provides the lowest memory consumption among the evaluated configurations.

\begin{table}[ht]
	\centering
	\caption{ESBT weight sequence path ($sc$) depth length under different insertion operation count.}
	\begin{tabular}{lcc}
		\hline
		\textbf{Operations} & \textbf{Average depth} & \textbf{Max depth} \\
		\hline
		\textbf{1\,000}   & 16.22 & 32  \\
		\textbf{10\,000}  & 161.37 & 323 \\
		\textbf{100\,000} & 1612.98 & 3226  \\
		\hline
	\end{tabular}
	\label{tab:esbt_middle_depth}
\end{table}

\noindent\textbf{\textit{Phase 2}}.
Table~\ref{tab:esbt_middle_depth} reports the effect of the maximum depth parameter on the length of the sequence path ($sc$). The results indicate that the average path length increases much more slowly than the number of insertions. Specifically, the average path length grows from 16.22 after 1\,000 insertions to 1\,612.98 after 100\,000 insertions, while the maximum observed path length reaches 3\,226 under the worst-case workload.\\
These results show that deep sequence paths occur only under extreme repeated insertions at the same logical position, whereas most identifiers remain relatively shallow in practice. Although the theoretical maximum depth observed during the experiment is 3\,226, a depth limit of 256 proved sufficient for all representative collaborative workloads considered in this study.
Consequently, the remaining experiments adopt the following parameter configuration:\\
\emph{base} = $2^{31}-1$ and \emph{depth} = 256,
which provides a favorable trade-off between memory consumption, identifier compactness, and allocation efficiency.
\subsection{ESBT versus Baseline Sequence CRDTs}
\label{EsbtvsCrdts}
This experiment compares the proposed ESBT with representative variable-size identifier sequence CRDTs, namely Logoot~\cite{logoot} and LSEQ~\cite{NedelecMMD13, NedelecMM21} using its Boundary, Random, and Mixed allocation strategies. The objective is to evaluate both execution efficiency and identifier compactness under identical collaborative editing workloads. Two performance metrics are considered:
\begin{itemize}
\item \textbf{Responsiveness}, measured as the cumulative time required for identifier generation and operation integration (ms).
\item \textbf{Identifier memory footprint}, measured as the total memory occupied by position identifiers (MB).
\end{itemize}
All measurements were averaged over multiple independent executions under identical JVM configurations to minimize runtime variability. 
The benchmark was implemented from scratch within a common Java framework. Although the implementations were inspired by the original Logoot~\cite{logoot}~\footnote{https://github.com/t-mullen/logoot-crdt} and LSEQ~\cite{NedelecMMD13}~\footnote{https://github.com/Chat-Wane/LSEQ} implementations, all algorithms were reimplemented using the same software architecture to ensure identical execution conditions and a fair comparison.

\subsubsection{Experimental Setup}
The benchmark simulates a collaborative editing session involving 50 replicas. One replica is designated as the local site and performs both local editing operations and the integration of remote updates. The remaining 49 replicas generate concurrent operations that are asynchronously propagated to the local replica. For each experiment, the total number of operations ranges from 10\,000 to 100\,000. To maximize concurrency, each remote replica generates an equal share of the workload (e.g., 2\,000 operations for a total workload of 100\,000 operations).
Two editing workloads are evaluated:
\begin{enumerate}
\item \textbf{Insertion-only}: 100\% insertions.
\item \textbf{Mixed editing}: 80\% insertions and 20\% deletions.
\end{enumerate}
For each workload, three insertion patterns are considered:
\begin{enumerate} 
	\vspace{-1mm}
	\item \textbf{Beginning:} All insertions occur at the start of the document. 
	\vspace{-2mm}
	\item \textbf{End:} All insertions occur document's end.
	\vspace{-2mm}
	\item \textbf{Random:} Positions are chosen uniformly at random across the document.
	\vspace{-2mm}
\end{enumerate}
\paragraph{Identifier Size Measurement.}
To ensure a fair comparison, memory measurements exclude the JVM object header and account only for the raw identifier fields. Under this model, an ESBT weight occupies 20 bytes, excluding the variable-length sequence path, whereas each Logoot/LSEQ identifier component occupies 16 bytes. Consequently, the total identifier size grows proportionally to the number of identifier components generated by each allocation strategy.

\paragraph{Worst-case Middle Insertion.}
The middle-insertion workload was evaluated using up to 10\,000 operations. Beyond approximately 55\,000 operations, the Logoot and LSEQ implementations exhausted the available JVM heap due to rapid identifier expansion. Restricting this workload ensures that all approaches can be compared under identical execution conditions. The impact of this behavior is further discussed in Sub-Section~\ref{middle-ins}.

\subsubsection{Pure (100\%) Insertion Workload Results}
\subsubsection*{a) Responsiveness}
Figure~\ref{fig:full_scenarios} compares the responsiveness of ESBT with Logoot and the three LSEQ allocation strategies under a pure insertion workload.\\
Under the \emph{Beginning} insertion pattern (Figure~\ref{subfig:res_begin_100}), all baseline approaches exhibit a steady increase in execution time as the number of operations increases. At 100\,000 insertions, Logoot and LSEQ-Random require approximately $940\,$ ms, reflecting the increasing cost of managing progressively longer identifiers. In contrast, ESBT requires only about $61\,$ ms, representing an execution time reduction of approximately 93\%. For the \emph{End} insertion pattern (Figure~\ref{subfig:res_end_100}), execution times decrease for all approaches because newly generated identifiers remain relatively compact. Nevertheless, Logoot and the LSEQ variants still require approximately $470\,$ ms after 100\,000 operations, whereas ESBT remains nearly constant at approximately $55\,$ ms, demonstrating stable allocation and integration performance. The \emph{Random} insertion pattern (Figure~\ref{subfig:res_rand_100}) produces the largest performance differences. Logoot and LSEQ-Random reach approximately 1\,000 ms due to increased identifier complexity under highly interleaved insertions. ESBT requires approximately $106\,$ ms, offering $10\times$ lower latency than Logoot and LSEQ variants.

\begin{figure*}[hbt!]
	\centering
	\subfigure[Beginning insertion pattern.]{
		\includegraphics[width=0.32\textwidth]{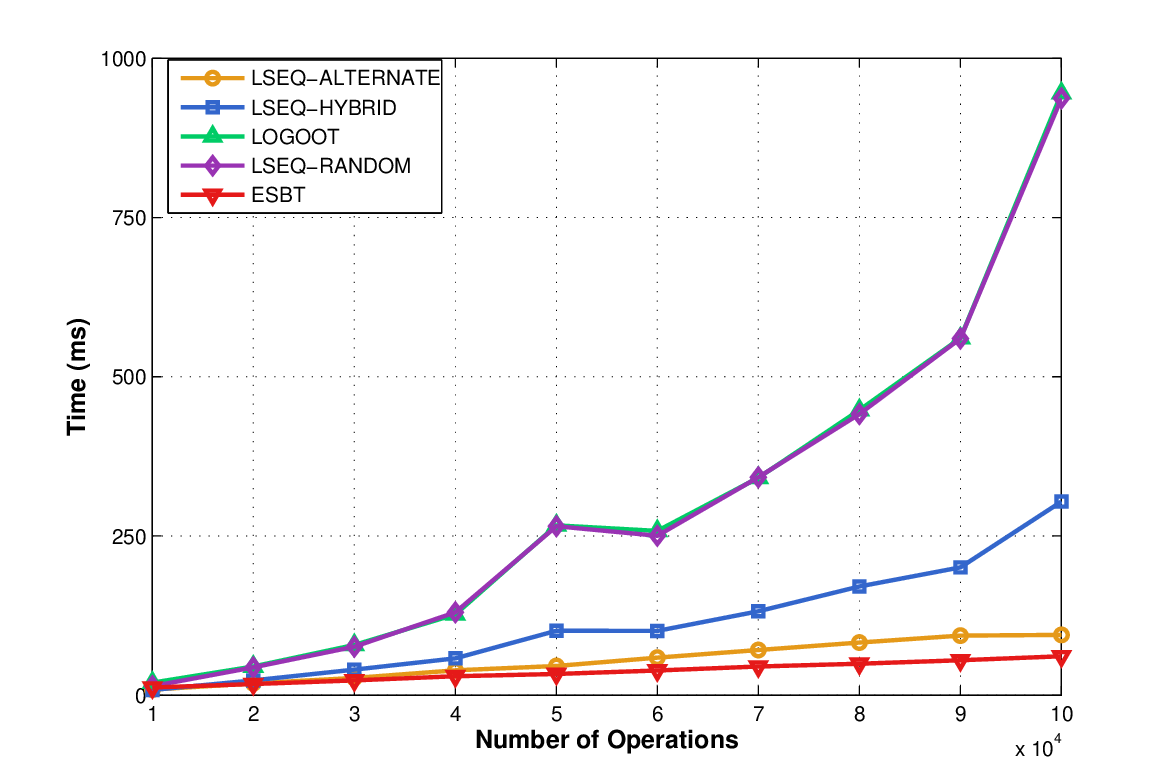}
		\label{subfig:res_begin_100}
	}
	\hspace{-0.3cm}
	\subfigure[End insertion pattern.]{
		\includegraphics[width=0.32\textwidth]{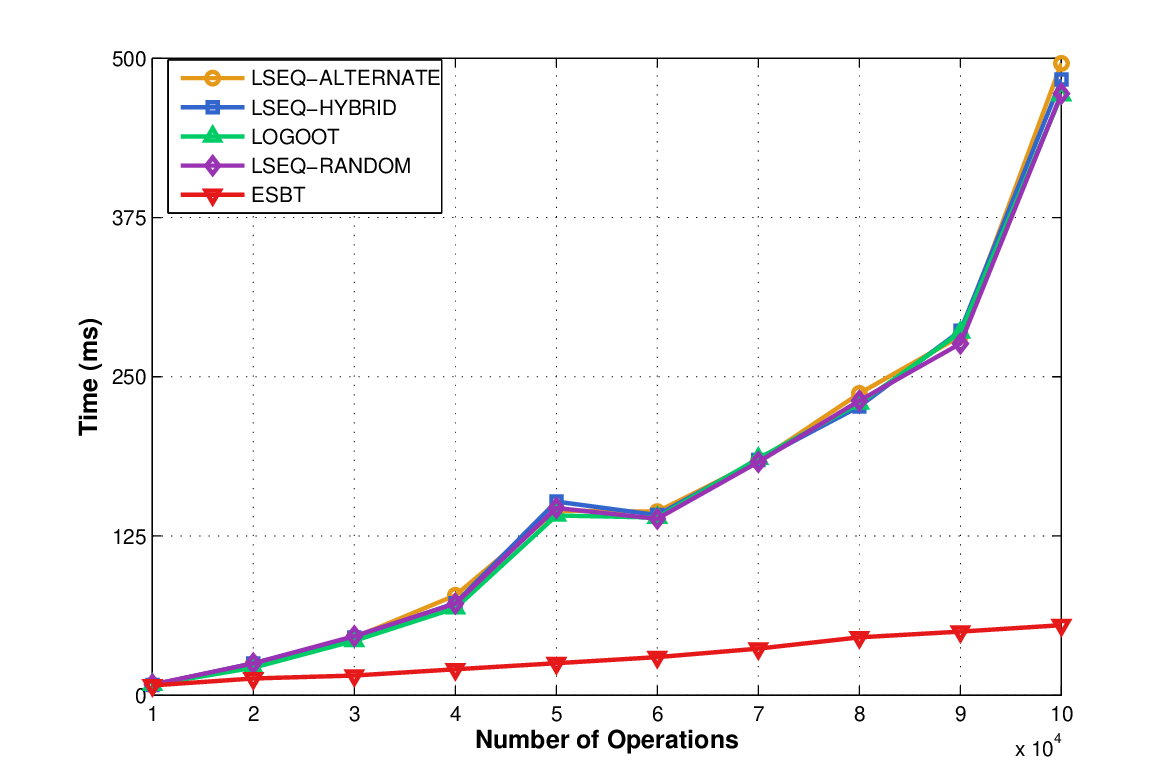}
		\label{subfig:res_end_100}
	}
	\hspace{-0.3cm}
	\subfigure[Random insertion pattern.]{
		\includegraphics[width=0.32\textwidth]{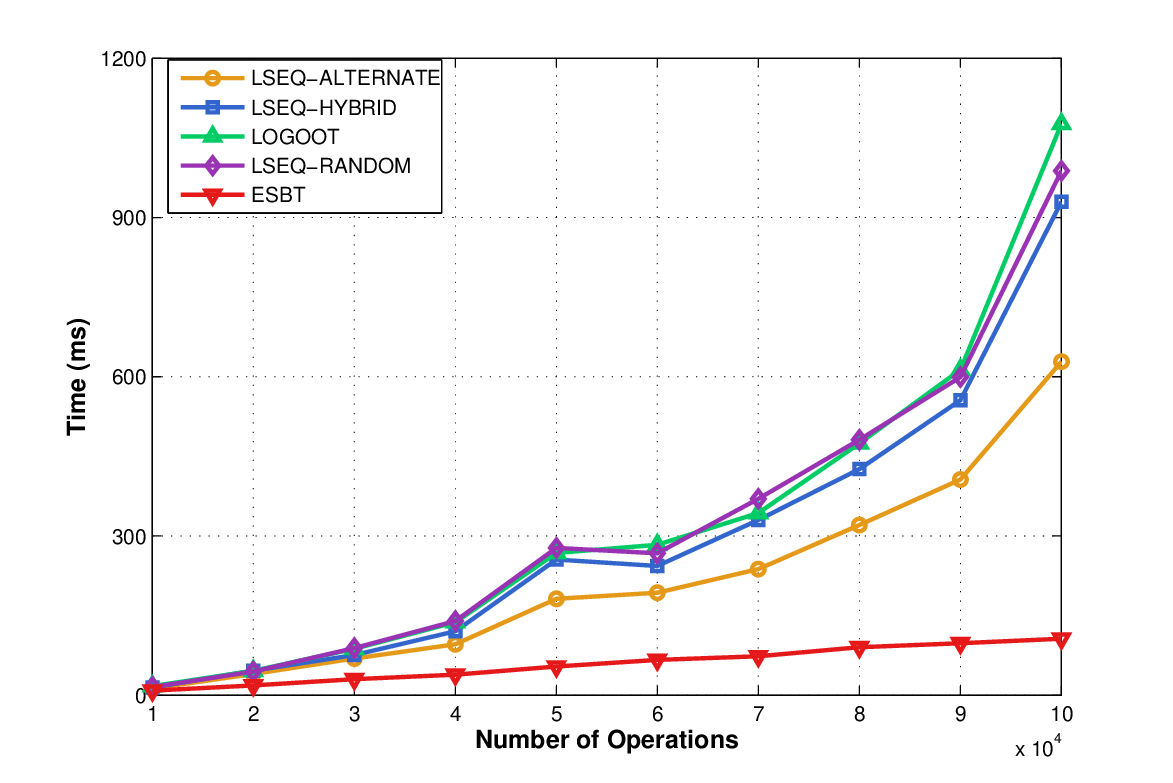}
		\label{subfig:res_rand_100}
	}
	
	\vspace{-0.2cm} 
	\caption{Responsiveness time (ms) under different insertion patterns and 100\% insertion workload.}
	\label{fig:full_scenarios}
\end{figure*}

\begin{figure*}[hbt!]
	\centering	
	\subfigure[Beginning insertion pattern.]{
		\includegraphics[width=0.32\textwidth]{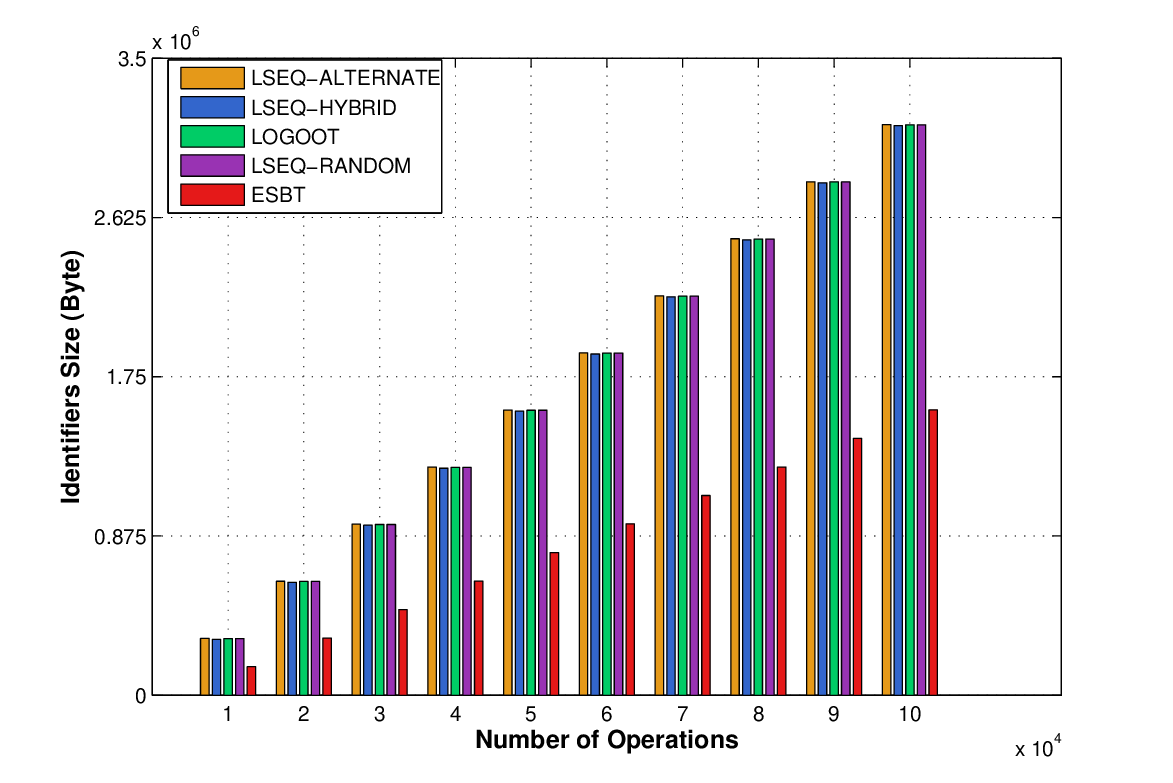}
		\label{subfig:byte_begin_100}
	}
	\hspace{-0.3cm}
	\subfigure[End insertion pattern]{
		\includegraphics[width=0.32\textwidth]{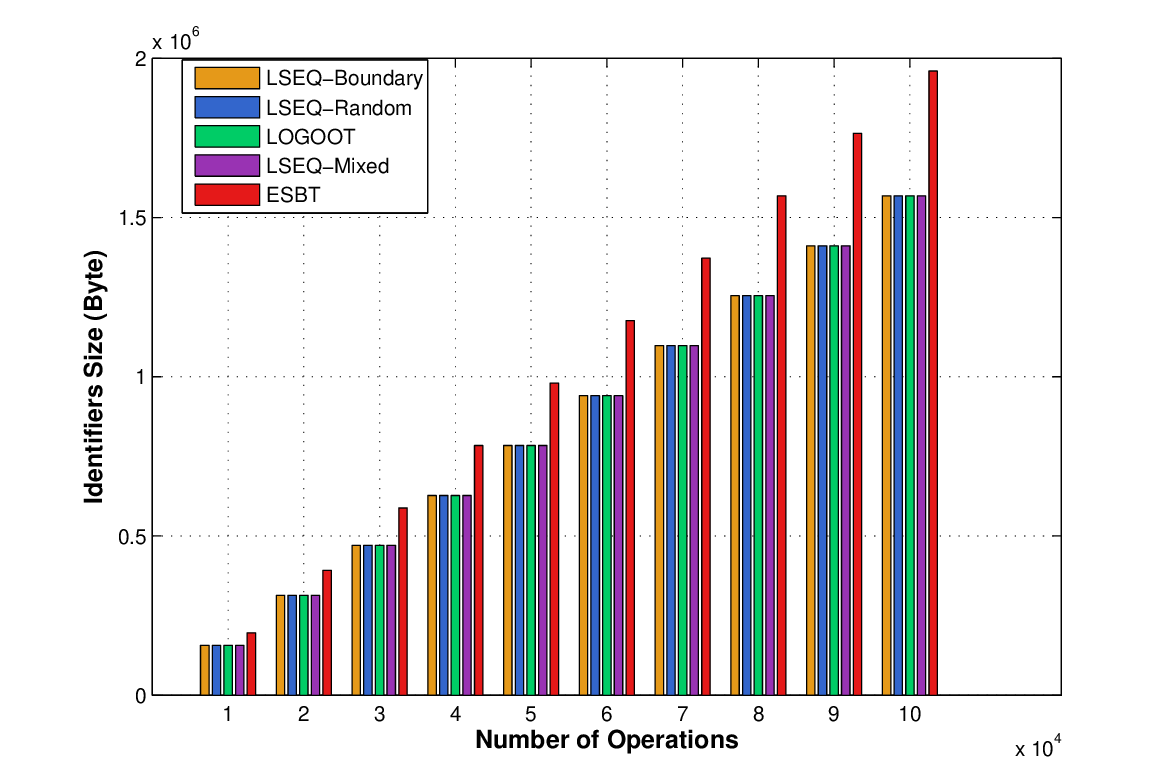}
		\label{subfig:byte_end_100}
	}
	\hspace{-0.3cm}
	\subfigure[Random insertion pattern]{
		\includegraphics[width=0.32\textwidth]{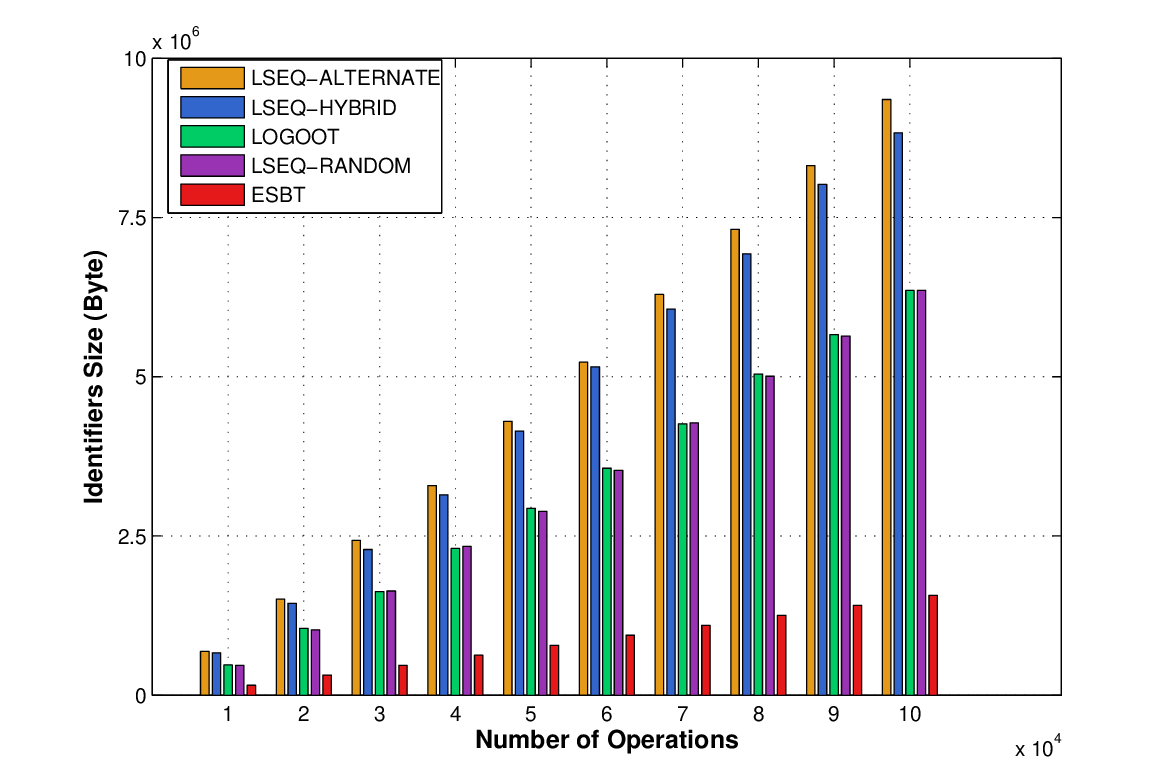}
		\label{subfig:byte_rand_100}
	}
	
	\vspace{-0.2cm}
	\caption{Identifier size (MB) under different insertion patterns and 100\% insertion workload.}
	\label{fig:full_byte_scenarios}
\end{figure*}

\subsubsection*{b) Identifier Memory Footprint}
Figure~\ref{fig:full_byte_scenarios} presents the total identifier memory footprint for the same workloads.\\
Under the \emph{Beginning} insertion pattern (Figure~\ref{subfig:byte_begin_100}), the memory consumption of Logoot and the LSEQ variants increases steadily as identifier lengths grow, reaching approximately $3.13$~MB after 100\,000 insertions. In comparison, ESBT increases from approximately $0.156$~MB to $1.56$~MB, maintaining a substantially smaller memory footprint throughout the experiment. For the \emph{End} insertion pattern (Figure~\ref{subfig:byte_end_100}), all approaches produce relatively compact identifiers because insertions occur at the document boundary. In this favorable scenario, ESBT identifiers are approximately 4 bytes larger than those of Logoot and LSEQ due to the additional metadata associated with the ESBT weight. However, this modest overhead is accompanied by significantly lower execution times, as shown in Figure~\ref{subfig:res_end_100}. The \emph{Random} insertion pattern (Figure~\ref{subfig:byte_rand_100} ) highlights the scalability of the different allocation strategies. Logoot and the LSEQ variants require between $6.3$ to $9.3$~MB after 100\,000 insertions, whereas ESBT remains close to $1.56$~MB. This corresponds to a reduction in identifier memory consumption of approximately 75 to 83\%, depending on the baseline algorithm.
\subsubsection{ Mixed (80\% Insert + 20\% Delete) Workload Results}
\subsubsection*{a) Responsiveness}
Figure~\ref{fig:resp_mixed_scenarios} reports the responsiveness of the evaluated approaches under a mixed workload consisting of 80\% insertions and 20\% deletions.
Introducing deletion operations reduces the number of active document elements and consequently lowers execution times for all approaches. Nevertheless, the relative performance trends remain unchanged.\\
Under the \emph{Beginning} insertion pattern (Figure~\ref{subfig:80_begin}), Logoot and the LSEQ variants require approximately $366\,$ ms after 100\,000 operations, whereas ESBT completes the same workload in approximately $64\,$ ms. For the \emph{End} insertion pattern (Figure~\ref{subfig:80_end}), the baseline approaches stabilize around $220\,$ ms owing to the reduced identifier complexity associated with end insertions. ESBT again exhibits the lowest execution time, remaining close to $56\,$ ms throughout the experiment. The \emph{Random} insertion pattern (Figure~\ref{subfig:80_random}) remains the most demanding scenario. Logoot, LSEQ-Random, and LSEQ-Mixed require approximately $400\,$ ms after 100\,000 operations, while ESBT requires only about $107\,$ ms, corresponding to execution-time reductions ranging from approximately 73\% to 89\%, depending on the baseline considered.\\
These results indicate that the performance advantages of ESBT remain consistent even when insertion and deletion operations are combined.

\begin{figure*}[hbt!]
	\centering
	\subfigure[Beginning 80\% insertion pattern.]{
		\includegraphics[width=0.32\textwidth]{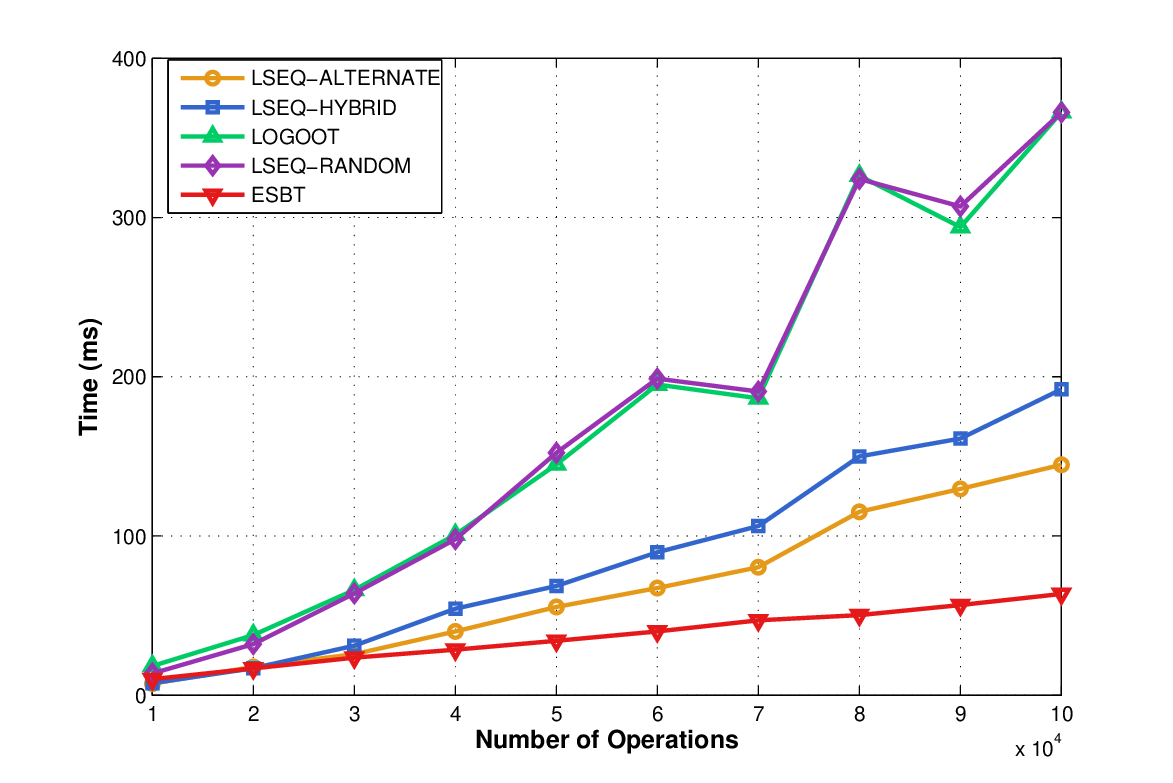}
		\label{subfig:80_begin}
	}
	\hspace{-0.3cm}
	\subfigure[End 80\% insertion pattern.]{
		\includegraphics[width=0.32\textwidth]{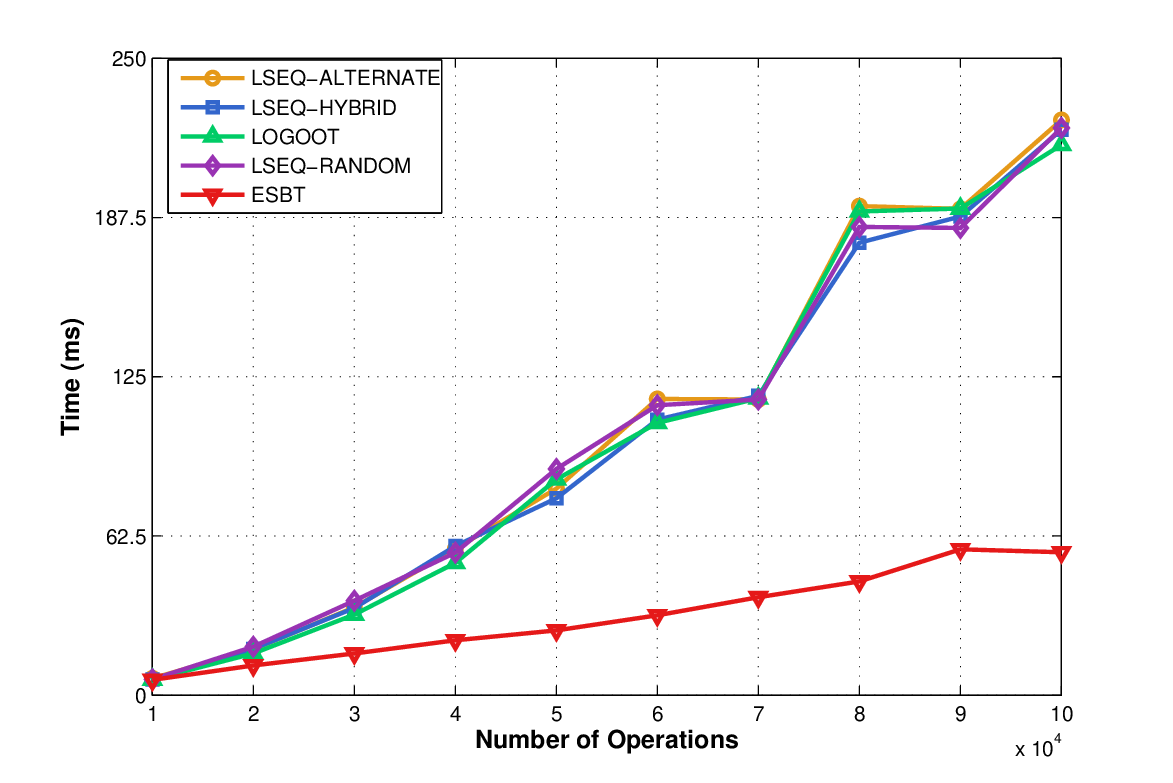}
		\label{subfig:80_end}
	}
	\hspace{-0.3cm}
	\subfigure[Random 80\% insertion pattern.]{
		\includegraphics[width=0.32\textwidth]{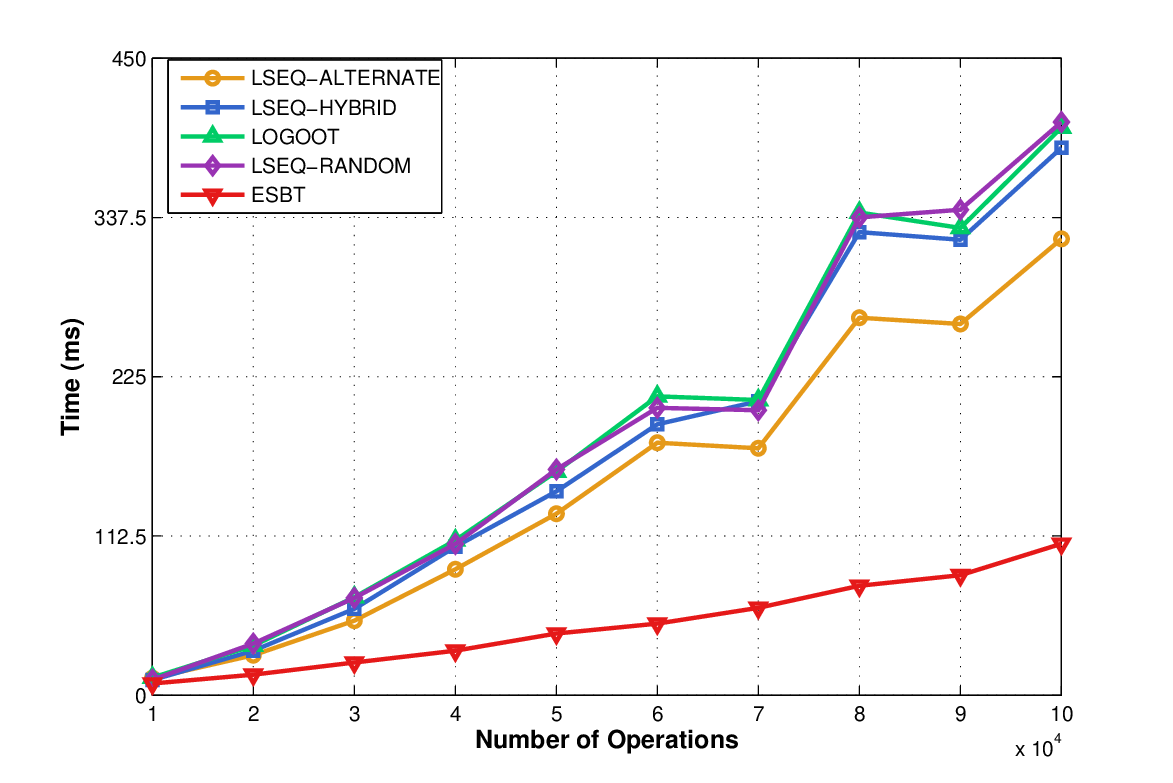}
		\label{subfig:80_random}
	}
	\vspace{-0.2cm}
	\caption{Responsiveness time (ms) under different insertion patterns and Mixed workloads.}
	\label{fig:resp_mixed_scenarios}
	
\end{figure*}
\vspace{-0.2cm} 

\begin{figure*}[hbt!]
	\subfigure[Beginning 80\% insertion pattern.]{
		\includegraphics[width=0.32\textwidth]{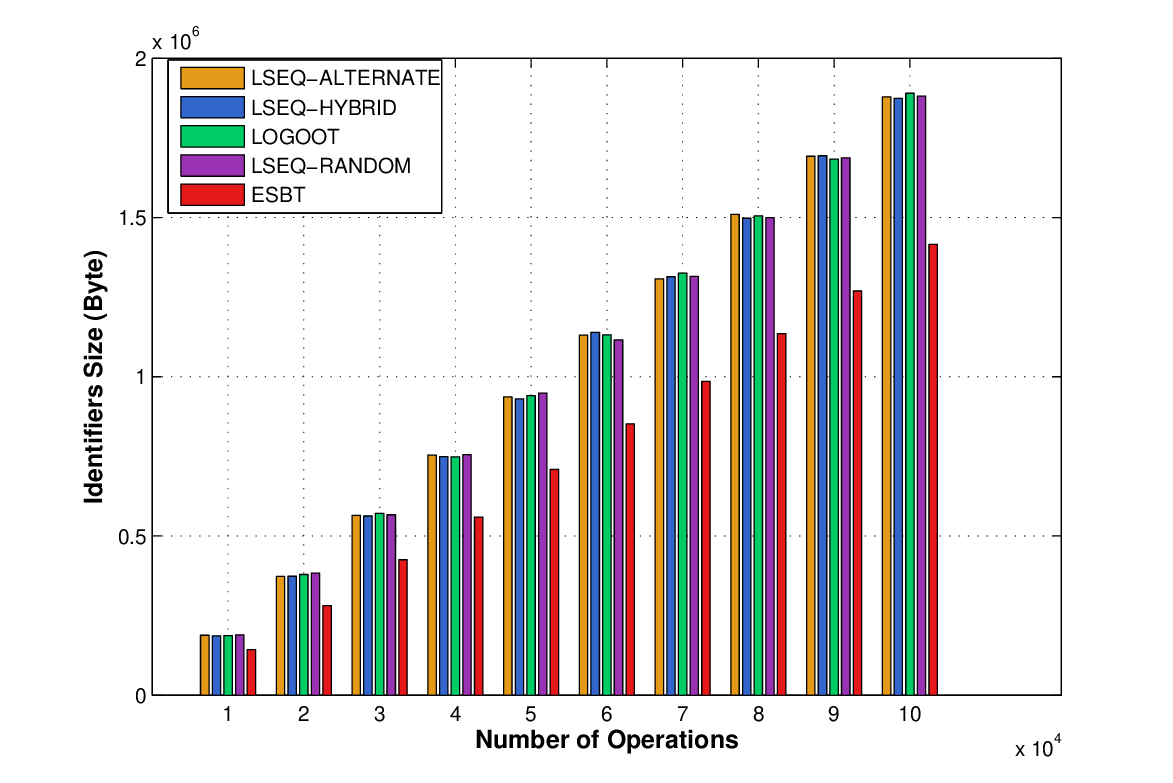}
		\label{subfig:byte80_beginning}
	}
	\hspace{-0.3cm}
	\subfigure[End 80\% insertion pattern.]{
		\includegraphics[width=0.32\textwidth]{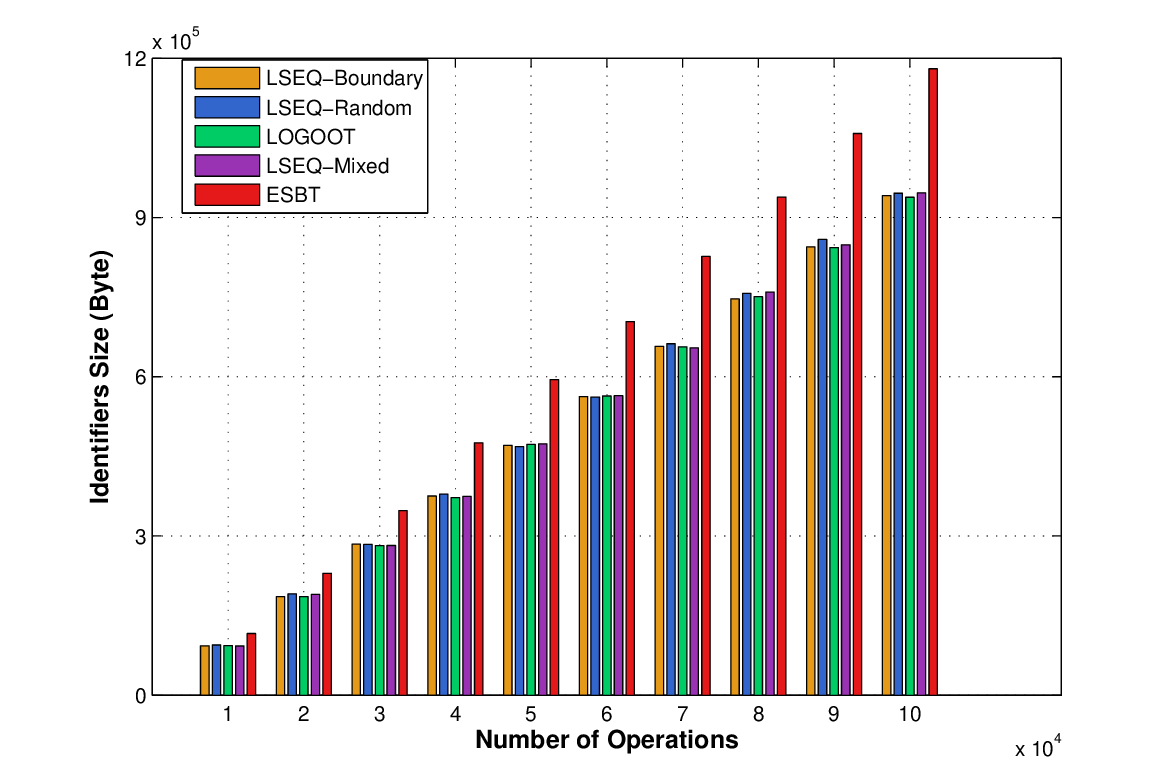}
		\label{subfig:byte80_end}
	}
	\hspace{-0.3cm}
	\subfigure[Random 80\% insertion pattern.]{
		\includegraphics[width=0.32\textwidth]{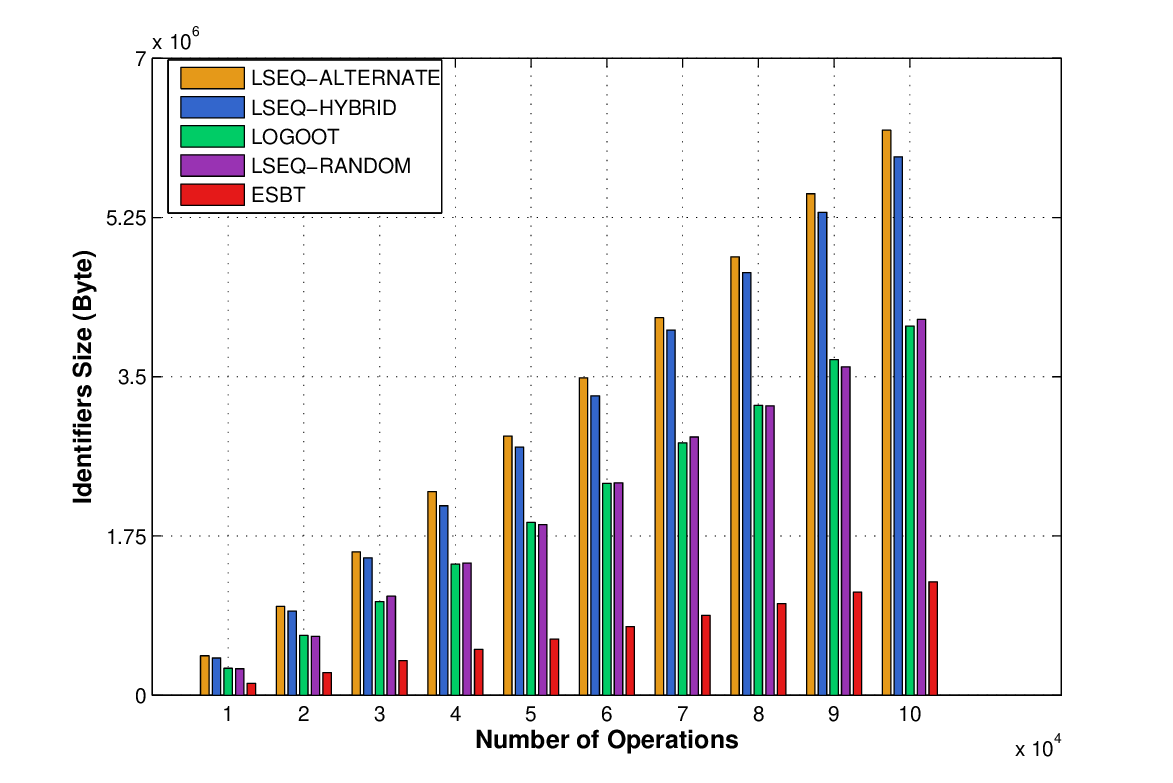}
		\label{subfig:byte80_random}
	}
	
	\vspace{-0.2cm}
	\caption{Identifier size (MB) under different insertion patterns and Mixed workload.}
	\label{fig:byte_mixed_scenarios}
\end{figure*}

\subsubsection*{b) Identifier Memory Footprint}
Figure~\ref{fig:byte_mixed_scenarios} compares the total identifier memory footprint under the mixed workload.
For the \emph{Beginning} insertion pattern (Figure~\ref{subfig:byte80_beginning}), Logoot and the LSEQ variants occupy between $1.8$ to $1.9$~MB after 100\,000 operations. ESBT requires approximately $937$~KB, reducing identifier memory consumption by about 50\%. 
Under the \emph{End} insertion pattern (Figure~\ref{subfig:byte80_end} ), all approaches maintain relatively compact identifiers. ESBT remains approximately 20\% larger than the baseline approaches because of its additional identifier metadata. However, this small increase in memory usage is accompanied by substantially lower execution times, as shown in Figure~\ref{subfig:80_end}. 
For the \emph{Random} insertion pattern (Figure~\ref{subfig:byte80_end}), identifier growth again becomes more pronounced for the baseline algorithms, whose memory footprint ranges from approximately $4.16$ to $6.20$~MB. ESBT remains close to $1.04$~MB, reducing identifier memory consumption by up to 79\%.
\subsection{Middle-Insertion pattern in pure insertion Workloads}
\label{middle-ins}
The middle-insertion workload represents the worst-case scenario for sequence CRDTs because each new identifier must be allocated strictly between two existing identifiers. This workload repeatedly refines the allocation interval and therefore stresses identifier allocation strategies more than the other insertion patterns.\\
For fairness, the evaluation was limited to 10\,000 operations, corresponding to the largest workload that all evaluated algorithms could complete.
During preliminary experiments, the baseline implementations were unable to complete larger workloads because of JVM heap exhaustion resulting from rapid identifier growth. Specifically, Logoot failed beyond approximately 30\,000 operations, LSEQ-Boundary and LSEQ-Mixed beyond approximately 40\,000 operations, and LSEQ-Random beyond approximately 55\,000 operations.

\subsubsection*{a) Responsiveness} 
Figure~\ref{fig:full_scenarios_Middle} compares the responsiveness of ESBT with the evaluated baseline approaches. Execution time increases rapidly for all baseline CRDTs as the number of middle insertions grows. At 10\,000 operations, Logoot and LSEQ-Random require approximately $380\,$ ms, whereas LSEQ-Boundary and LSEQ-Mixed exceed $700\,$ ms. This behavior is consistent with the increasing complexity of identifier allocation and comparison under repeated midpoint insertions. In contrast, ESBT increases from approximately $5\,$ ms at 1\,000 operations to only $52\,$ ms at 10\,000 operations. Compared with the baseline approaches, ESBT reduces execution time by approximately 86 to 93\%, corresponding to a speedup ranging from $7\times$ to $14\times$, depending on the allocation strategy.

\subsubsection*{b) Identifier Memory Footprint}
Figure~\ref{subfig:ByteFullMiddle} reports the total identifier memory footprint under the same workload. 
Repeated middle insertions produce rapid identifier growth in all baseline approaches. At 10\,000 operations, LSEQ-Boundary and LSEQ-Mixed require more than $8$~MB of identifier storage, whereas Logoot and LSEQ-Random require approximately $4.2$~MB. These results explain why the baseline implementations were unable to complete larger workloads within the available JVM heap. In comparison, ESBT maintains a substantially smaller memory footprint, increasing from approximately $0,015$~MB at 1\,000 operations to $0,298$~MB at 10\,000 operations. This corresponds to a reduction in identifier memory consumption ranging from approximately 93\% to 96\% (or $13\times$ to $27\times$ smaller) relative to the evaluated baseline approaches. 
These results demonstrate that the bounded identifier allocation strategy of ESBT effectively limits identifier growth under repeated midpoint insertions, making it well suited for highly adversarial collaborative editing workloads.
\begin{figure*}[hbt!]
	\centering
	\subfigure[Responsiveness time under 100\% middle-position insertions]{
		\includegraphics[width=0.32\textwidth]{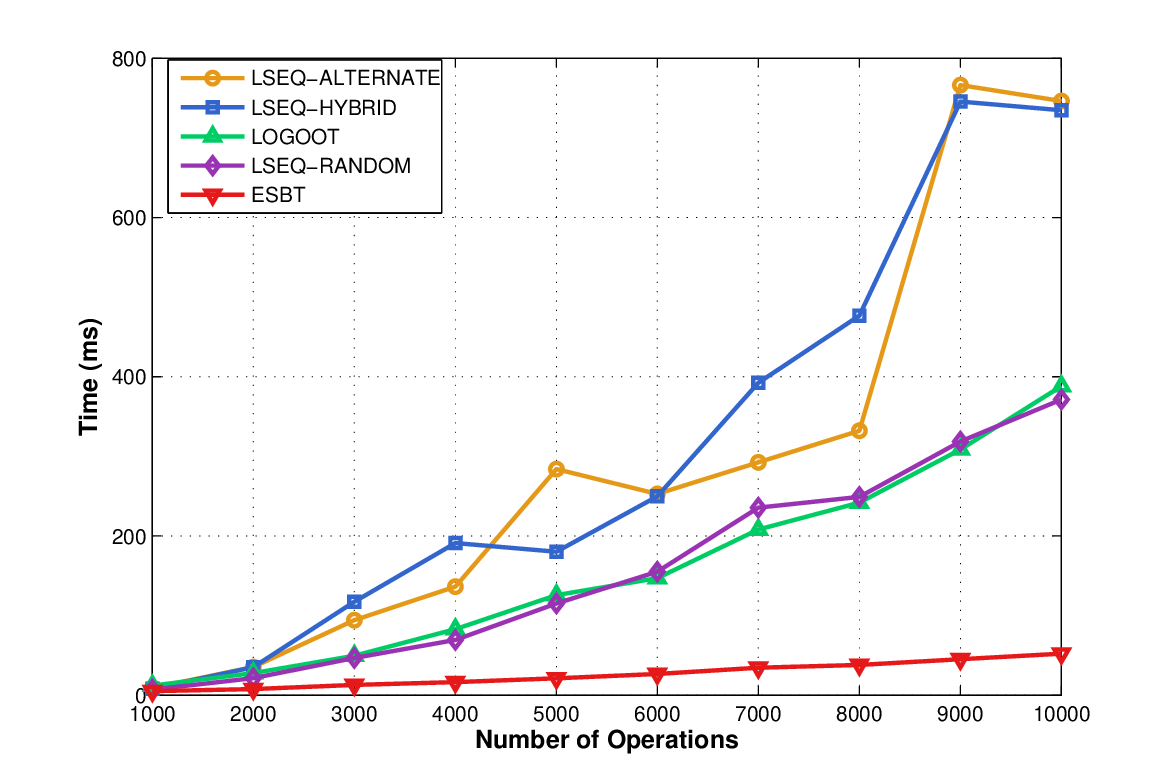}
		\label{subfig:fullMiddle_begin}
	}
	\hspace{-0.3mm}
	\subfigure[Identifier size under 100\% middle-position insertions]{
		\includegraphics[width=0.32\textwidth]{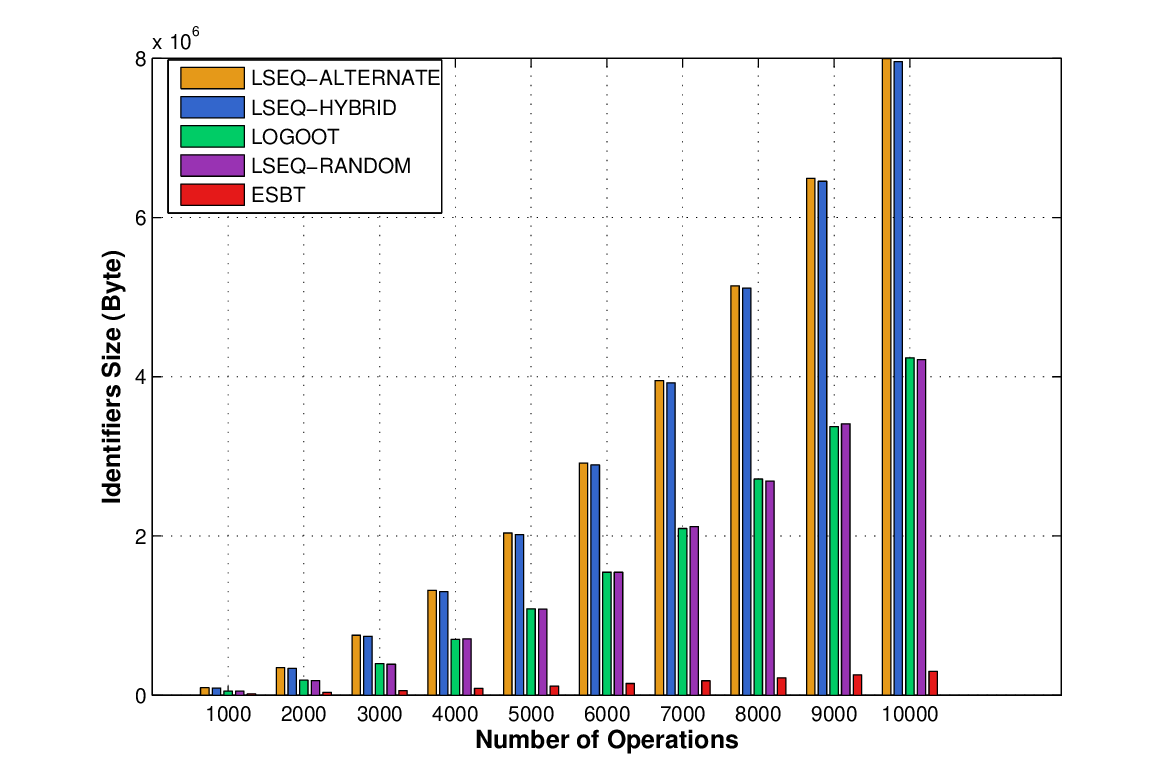}
		\label{subfig:ByteFullMiddle}
	}
	
	\caption{ESBT and baseline CRDTs performance under a middle-insertion pattern with pure insertion workloads.}
	\label{fig:full_scenarios_Middle}
\end{figure*}


\subsubsection{Discussion}
\label{subsec:discussion}
The experimental results consistently show that ESBT achieves lower execution time and smaller identifier memory footprints than the evaluated sequence CRDTs across all workloads. The largest improvements are observed under the Beginning, Random, and Middle insertion patterns, where identifier growth becomes the dominant factor affecting the performance of existing variable-size identifier allocation strategies.\\
The observed performance improvements originate from two complementary design choices. First, ESBT bounds the growth of the fraction component through the parameter $D_{\max}$ and delegates further disambiguation to successive identifier layers only when necessary. Consequently, identifier expansion remains proportional to the actual level of contention rather than to the number or locality of insertions. Second, storing document elements in a Red--Black tree preserves logarithmic-time $O(\log n)$ search, insertion, and deletion, allowing the execution time to remain stable even as the document size increases.\\ 
Under the End insertion workload, ESBT exhibits a modest increase in identifier size compared with the best-performing baseline. This difference is primarily due to the current prototype implementation, which explicitly stores the sequence-path field even when it contains only its default value ($sc=[0]$). Since this default value can be represented implicitly during serialization and reconstructed during deserialization, the additional storage cost is implementation-specific rather than intrinsic to the ESBT allocation strategy. In contrast, the growth observed in Logoot and LSEQ identifiers is an inherent property of their allocation mechanisms. As concurrent insertions repeatedly target the same allocation interval, additional identifier components are introduced to preserve ordering, progressively increasing both memory consumption and comparison costs. This behavior becomes particularly pronounced under adversarial middle-insertion workload, where ESBT reduces execution time by up to 86.5\% and identifier memory consumption by up to 92.8\% relative to the evaluated baselines.
Finally, ESBT can be naturally extended to support very large-scale collaborative editing environments. The current implementation of fraction $f$ uses 32-bit integers, which are sufficient for the workloads evaluated in this paper. Extending these fraction fields to 64-bit integers would substantially enlarge the available identifier space, allowing much deeper mediant refinement before reaching the 
$D_{\max}$ bound, without modifying the ESBT allocation algorithm, its correctness guarantees, or its asymptotic complexity.

\section{Related Works}
\label{sec:related_work}
Several sequence CRDTs have been proposed for collaborative editing. Each aims to maintain a consistent order of elements inserted simultaneously while avoiding centralized coordination. This section reviews representative approaches and discusses their respective strengths and limitations.\\
WOOT (WithOut Operational Transforms)~\cite{wooto} is considered as the first based-CRDT algorithm proposed in literature. In WOOT, each element of a text document is associated with a unique structure comprising its identifier, its content, and the identifiers of its predecessor and successor elements. WOOT uses a monotonic linearization function to maintain convergence between sites, but does not allow the deletion of elements, which introduces additional memory and communication overhead as deleted elements accumulate over time.
To address this issue, WOOT uses the tombstone technique~\cite{crdt1}, where deleted elements are retained but made invisible to users. However, this approach can lead to memory and bandwidth overhead problems~\cite{sec17}. 
Improved versions WOOTO~\cite{wooto} and WOOTH~\cite{wooth} have been proposed, but they only optimize integration complexity from $O(n^3)$ to $O(n^2)$ where $n$ denotes the number of elements. 
In~\cite{woot22}, the authors provide a formal proof that WOOT achieves strong eventual consistency and propose size-optimized messages for update operations using a sort key-based protocol.

TreeDoc~\cite{treedoc} is a CRDT that uses a binary tree to represent the document. Each node of the tree contains an element and two children (left and right). Each node has a unique identifier which is its path in the tree. The content of node is represented by the infix path of the tree. However, when many insertions occur near the end of the document, which is the normal editing situation, repeated localized insertions may increase tree depth and identifier length.
TreeDoc uses tombstone as it cannot delete a node that already has a children. But it can safely delete nodes that does not have any visible children. Moreover, maintaining balance and reclaiming tombstones may incur additional synchronization overhead in large-scale collaborative environments.

Logoot~\cite{logoot} is a CRDT that considers a document as a set of lines all identified by absolute positions totally order using lexicographic order. These identifiers are only used once, meaning they cannot be reused. Logoot document consists of lines defined as \textit{$<$pid, content$>$} where the \textit{content} is a line of text and the \textit{pid} is a unique position identifier. Logoot does not use tombstone for hiding deleted elements, it removes physically the elements form its document. However, when several insertions are applied in the same part of the document, the generated identifiers for these insertions progressively exhaust the available identifier space. This leads to increasingly long identifiers and higher comparison costs.
  
LSEQ~\cite{NedelecMMD13, NedelecMM21} refines Logoot's identifier allocation by using different identifier allocation strategies adapted to different editing patterns. LSEQ combines two classical strategies: a Random allocator, which disperses identifiers uniformly within available intervals, and a Boundary allocator, which selects positions at interval extremities. While Random strategy enhances distribution fairness, Boundary strategy provides predictability at the cost of accelerated identifiers growth under high concurrency. LSEQ introduces an adaptive round-robin policy that alternates between these strategies to balance dispersion and stability while reducing identifier depth compared to either strategy alone. 
Despite these improvements, LSEQ still suffers from unbounded identifier growth under intense editing, path expansion in highly concurrent regions, and workload sensitivity, since boundary-heavy scenarios or repeated midpoint insertions can still produce long identifier paths that degrade memory footprint and integration responsiveness. Also, LSEQ does not ensure deterministic and bounded identifier allocation since it relies on several antagonist allocation strategies. Therefore, without coordination among collaborators, identifiers can grow quadratically in size~\cite{NedelecMMD13a}.

\begin{table*}[!ht]
\centering
\caption{Comparison of representative sequence CRDT algorithms.}
\label{com_crdt}

\begin{tabular}{
l
>{\centering\arraybackslash}p{1.5cm}
>{\centering\arraybackslash}p{1.5cm}
>{\centering\arraybackslash}p{2.4cm}
>{\centering\arraybackslash}p{1.7cm}
>{\centering\arraybackslash}p{1.6cm}
>{\centering\arraybackslash}p{2.2cm}
>{\centering\arraybackslash}p{1.6cm}
}
\toprule
\multirow{2}{*}{\textbf{Algorithm}} &
\multicolumn{2}{c}{\textbf{Complexity}} &
\multirow{2}{*}{\textbf{Data Structure}} &
\multirow{2}{*}{\begin{tabular}[c]{@{}c@{}}\textbf{Tombstone}\\ \textbf{Based}\end{tabular}} &
\multirow{2}{*}{\begin{tabular}[c]{@{}c@{}}\textbf{GC}\\ \textbf{Required}\end{tabular}} &
\multirow{2}{*}{\begin{tabular}[c]{@{}c@{}}\textbf{Identifier}\\ \textbf{Size}\end{tabular}} &
\multirow{2}{*}{\begin{tabular}[c]{@{}c@{}}\textbf{Metadata}\\ \textbf{Overhead}\end{tabular}} \\
\cmidrule(lr){2-3}
& \textbf{Local} & \textbf{Remote} & & & & & \\
\midrule

WOOT~\cite{woot}
& $O(n)$
& $O(n^2)$
& Array List
& Yes
& Yes
& Fixed
& High \\

TreeDoc~\cite{treedoc}
& $O(n)$
& $O(n\log n)$
& Extended Binary Tree
& Yes
& Yes
& Variable
& High \\

Logoot~\cite{logoot}
& $O(n)$
& $O(n\log n)$
& Array List
& No
& No
& Variable
& High \\

LSEQ~\cite{NedelecMMD13}
& $O(n)$
& $O(n\log n)$
& Array List
& No
& No
& Variable
& Medium \\

RGA~\cite{rga}
& $O(n)$
& $O(\log n)$
& Linked List
& Yes
& Yes
& Fixed
& High \\

Eg-Walker~\cite{Eg-walker}
& $O(\log n)$
& $O(\log n)$
& B-Tree
& Yes
& Yes
& Fixed
& Low \\

\textbf{ESBT}
& $\mathbf{O(\log n)}$
& $\mathbf{O(\log n)}$
& \textbf{Red-Black Tree}
& \textbf{No}
& \textbf{No}
& \textbf{Variable}
& \textbf{Low} \\

\bottomrule
\end{tabular}

\vspace{2mm}
\footnotesize
\textit{Metadata overhead} qualitatively reflects the amount of auxiliary information maintained by each algorithm, including tombstones, identifier expansion, and synchronization metadata. \textit{Identifier size} characterizes the evolution of position identifiers
: \emph{Fixed} denotes constant-size identifiers, whereas \emph{Variable} denotes identifiers whose size may increase depending on insertion patterns and editing workloads.
\end{table*}
RGA (Replicated Growable Array)~\cite{rga} is a linked-list
data ordered consistently with the happened-before relation. Each element is assigned a unique timestamp-based identifier (s4vector).  RGA uses a hash table maps identifiers to document elements for efficient lookup and conflict resolution. RGA uses tombstones to preserve ordering after deletions, and its correctness has been formally established in~\cite{sec17}. However, its message size grows with the number of replicas and update operations. Several extensions, including PRGA~\cite{prga} and Fugue~\cite{fugue}, have been proposed to improve performance, But still rely on tombstones, incurring metadata and garbage-collection overhead.\\
Eg-walker~\cite{Eg-walker} introduces a hybrid approach that combines principles of Operational Transformation (OT) and Conflict-free Replicated Data Types (CRDTs) that embeds editing updates operations as a directed acyclic graph (DAG) to maintain causal relations among operations. Eg-Walker constructs a causality-preserving walk over the DAG to derive a consistent total order, eliminating the need for variable-length identifiers as used in Logoot or LSEQ. However, Eg-Walker inherits several limitations associated with DAG-based representations including unbounded DAG growth, complex garbage collection, expensive graph traversal during integration, and high memory consumption under intense concurrent workloads which makes it less suitable for large-scale and long-lived collaborative documents.\\
Table~\ref{com_crdt} summarizes the main characteristics of representative sequence CRDTs, comparing their computational complexity, underlying data structures, tombstone management, garbage collection (GC) requirement, identifier behavior and metadata overhead. Existing approaches can generally be categorized according to their ordering strategy. Tombstone-based CRDTs maintain deleted elements to preserve ordering, which increases metadata size and may require garbage collection (GC) to reclaim obsolete state. For example, WOOT stores predecessor and successor references together with tombstones for each deleted element. RGA also maintains tombstones and additional timestamp-based identifiers to preserve causal ordering. In contrast, identifier-allocation approaches (e.g., Logoot and LSEQ) eliminate tombstones but may suffer from identifier growth under unfavorable insertion patterns, increasing both memory consumption and comparison costs. These pathological insertion patterns may also be deliberately generated by an adversarial participant to inflate identifier size and degrade system performance. ESBT maintains low metadata overhead by storing only the weight and a lightweight synchronization counter, without tombstones or auxiliary ordering structures. ESBT belongs to the second category but introduces a bounded Stern–Brocot-based identifier allocation strategy with logarithmic-time operations through its integration with a self-balanced Red–Black tree.
The effectiveness of these design choices is quantitatively evaluated in Section~\ref{sec:performance} through execution time and memory consumption experiments under representative collaborative editing workloads.

\section{Conclusion}
\label{sec:conclusion}
In this paper, we proposed ESBT, a novel sequence CRDT that addresses the scalability limitations of existing identifier-allocation approaches for collaborative editing. ESBT combines deterministic identifier allocation with efficient synchronization, providing bounded identifier growth, tombstone-free deletion management, and logarithmic-time document operations.
Experimental results demonstrate that ESBT consistently reduces identifier memory consumption and execution time compared with representative sequence CRDTs, including Logoot~\cite{logoot} and LSEQ~\cite{NedelecMMD13, NedelecMM21}, across diverse collaborative editing workloads. The proposed approach also maintains stable performance under challenging insertion patterns while preserving deterministic ordering and Strong Eventual Consistency.\\
Future work will focus on four directions. First, we will investigate adaptive mechanisms for automatically tuning the 
$D_{\max}$ parameter according to observed editing dynamics and workload characteristics. Second, we will explore compact encoding and serialization techniques for the sequence path to further reduce identifier storage requirements. Third, we plan to integrate ESBT as a pluggable identifier allocation strategy into existing collaborative editing frameworks, such as Yjs\footnote{https://docs.yjs.dev/api/internals} and Automerge\footnote{https://automerge.org/}, and evaluate its effectiveness under real-world collaborative workloads. Finally, we will investigate the behavior and performance of ESBT under adverse network conditions, including network partitions, prolonged disconnections, and replica recovery, to further assess its robustness in large-scale distributed collaborative environments.

\bibliographystyle{unsrtnat} 
\bibliography{mybib}

\end{document}